\documentclass[12pt,draftcls,journal,onecolumn]{IEEEtran}
\usepackage{amsfonts,color,morefloats,pslatex,graphicx,graphics}
\usepackage{amssymb,amsthm, amsmath,latexsym}

\newcommand{\red}{\color{red}}
\newtheorem{theorem}{Theorem}
\newtheorem{lemma}[theorem]{Lemma}
\newtheorem{remark}[theorem]{Remark}
\newtheorem{proposition}[theorem]{Proposition}

\newtheorem{example}[theorem]{Example}

\newcommand{\ord}{{\mathrm{ord}}}

\newcommand{\lcm}{{\mathrm{lcm}}}

\newcommand{\gf}{{\mathrm{GF}}}

\newcommand{\C}{{\mathcal{C}}}

\usepackage{blindtext}

\ifCLASSINFOpdf

\else

\fi

\begin{document}
%
\title{Two classes of narrow-sense BCH codes and their duals
\thanks{The work of Chengju Li was supported by the National
Natural Science Foundation of China (12071138), Shanghai Rising-Star Program (22QA1403200), the open research fund of National Mobile Communications Research Laboratory of Southeast University (2022D05), and the Shanghai Trusted Industry Internet Software Collaborative Innovation Center.}
}

\author{Xiaoqiang Wang, Jiaojiao Wang, Chengju Li*, Yansheng Wu{\thanks{Corresponding author*.
\newline \indent ~~Xiaoqiang Wang is with Hubei Key Laboratory of Applied Mathematics, Faculty of Mathematics and Statistics, Hubei University, Wuhan 430062, China (E-mail: waxiqq@163.com).
\newline \indent ~~ Jiaojiao Wang is with Data Science and Information Technology Research Center, Tsinghua-Berkeley Shenzhen Institute, Tsinghua Shenzhen Inter national Graduate School, Shenzhen, China. (E-mail:  wjj22@mails.tsinghua.edu.cn).
\newline \indent ~~ Chengju Li is the Shanghai Key Laboratory of Trustworthy Computing, East China Normal University,
Shanghai, 200062, China; and is also with the National Mobile Communications Research Laboratory, Southeast University, Nanjing 210096, China  (E-mail:  cjli@sei.ecnu.edu.cn).
\newline \indent ~~Yansheng Wu is with School of Computer Science, Nanjing University of Posts and Telecommunications, Nanjing 210023, China (E-mail: yanshengwu@njupt.edu.cn).}}

}

\maketitle

\begin{abstract}
BCH codes and their dual codes are two special subclasses of cyclic codes and are the best linear codes in many cases. A lot of progress on the study
of BCH cyclic codes has been made, but little is known about the minimum distances of the duals of BCH codes.  Recently, a new concept called dually-BCH code was introduced to
investigate the duals of BCH codes and the lower bounds on their minimum distances in \cite{GDL21}. For a prime power $q$ and an integer $m \ge 4$, let $n=\frac{q^m-1}{q+1}$ \ ($m$ even), or $n=\frac{q^m-1}{q-1}$ \ ($q>2$).
 In this paper, some sufficient and necessary conditions in terms of the designed distance will be given to ensure that the narrow-sense BCH codes of length $n$ are dually-BCH codes, which extended the results in \cite{GDL21}. Lower bounds on the minimum distances of their dual codes are developed for $n=\frac{q^m-1}{q+1}$ \ ($m$ even). As byproducts, we present the largest coset leader $\delta_1$ modulo $n$ being of two types, which proves a conjecture in \cite{WLP19} and partially solves an open problem in \cite{Li2017}. We also investigate the parameters of the narrow-sense BCH codes of length $n$ with design distance $\delta_1$. The BCH codes presented in this paper have good parameters in general.
\end{abstract}


%
\IEEEpeerreviewmaketitle

\section{Introduction}\label{sec-auxiliary}

Let $\gf(q)$ be the finite field with $q$ elements, where $q$ is a prime power. Let $n$, $k$ be two positive integers such that $k\leq n$. An $[n, k, d]$ linear code $\mathcal{C}$ over the finite field $\gf(q)$ is a $k$-dimensional linear subspace of $\gf(q)^n$ with minimum (Hamming) distance $d$. The dual code of $\mathcal{C}$, denoted by $\mathcal{C}^{\perp}$, is defined by
$$\mathcal{C}^{\perp}=\{\mathbf{b} \in \gf(q)^n\,:\,\mathbf{b}\mathbf{c}^{T}=0 \,\,\text{for any $\mathbf{c} \in \mathcal{C}$}\}, $$
where $\mathbf{b}\mathbf{c}^{T}$ is the standard inner product of two vector $\mathbf{b}$ and $\mathbf{c}$. If the code $\mathcal{C}$ is closed under the cyclic shift, i.e., if
$(c_0,c_1, c_2, \cdots, c_{n-1}) \in \C$ implies $(c_{n-1}, c_0, c_1, c_2,\cdots, c_{n-2})
\in \C$, then $\mathcal{C}$ is called a {\it cyclic code}.
 By identifying any vector $(c_0,c_1,\cdots,c_{n-1}) \in \gf(q)^n$ corresponds to a polynomial
$$c_0+c_1x+c_2x^2+\cdots+c_{n-1}x^{n-1} \in \gf(q)[x]/(x^n-1),$$
then any cyclic code $\C$ of length $n$ over $\gf(q)$ corresponds to a subset of the quotient ring
$\gf(q)[x]/(x^n-1)$.
  Since every ideal of $\gf(q)[x]/(x^n-1)$ must be principal, $\C$ can be expressed as $\C=\langle g(x)\rangle$, where
$g(x)$ is a monic polynomial with the smallest degree and is called the {\it generator polynomial}. Let $h(x)=(x^n-1)/g(x)$, then $h(x)$ is referred to as the {\em check polynomial} of $\C$.
The zeros of $g(x)$ and $h(x)$ are called zeros and non-zeros of $\C$.

 An $[n, k, d]$ linear
code over $\gf(q)$ is said to be \textit{distance-optimal} (respectively,
\textit{dimension-optimal} and \textit{length-optimal}) if there does not exist $[n, k, d' \geq d+1]$ (respectively,
$[n, k' \geq k+1, d]$ and $[n' \leq n-1, k, d]$) linear code over $\gf(q)$. A code is called optimal code if it is
 length-optimal, or dimension-optimal, or distance-optimal, or meets a bound for linear codes.

Let $n$ be a positive integer with $\gcd(n,q)=1$. Let $\ell=\ord_{n}(q)$ be the order of $q$ modulo $n$, and
let $\alpha$ be a generator of the group $\gf(q^\ell)^*$. Put $\gamma=\alpha^{(q^\ell-1)/n}$,
then $\gamma$ is a primitive $n$-th root of unity in the finite field $\gf(q^\ell)$. For any $i$ with $0\leq i\leq q^\ell-2$,
let $m_i(x)$ denote the minimal polynomial of $\gamma^i$ over $\gf(q)$. For any $\delta$
 with $2 \leq \delta \leq n$, let
\begin{eqnarray*}
g_{(q,n,\delta,b)}=\lcm \left(m_{b}(x), m_{b+1}(x), \cdots, m_{b+\delta-2}(x)\right),
\end{eqnarray*}
where $b$ is an integer and lcm denotes the least common multiple of these minimal polynomials. Let $\C_{(q,n,\delta,b)}$ denote
the cyclic code of length $n$ with generator polynomial $g_{(q,n,\delta,b)}$, then $\C_{(q,n,\delta,b)}$ is called a
 BCH code with {\it designed distance} $\delta$. If $b=1$, then $\C_{(q,n,\delta,1)}$ is called a {\it narrow-sense BCH code}. In this case, for convenience we will use $\mathcal{C}_{(q,n,\delta)}$ in the sequel. A BCH code is called {\it dually-BCH code} if both the BCH code and its dual are a BCH code with respect to an $n$-th primitive root of unity $\gamma$. This concept was introduced in \cite{GDL21} to
investigate the duals of BCH codes and the lower bounds on their minimum distances.

BCH codes were invented in 1959 by Hocquenghem \cite{Hocquenghem59}, and independently in 1960  Bose and Ray-Chaudhuri \cite{Bose62}. They were extended to BCH
codes over finite fields by Gorenstein and Zierler in 1961 \cite{Gorenstein61}. In the past several decades,  BCH codes have been widely studied and are treated in almost every book on coding theory since
they are a special class of cyclic codes with interesting properties and applications and are usually among the best cyclic codes.
The reader is referred to, for example, \cite{Aly07,Augot94,Charpin90,Ding2015,Ding15,Ding17,Lid017,Lid17,Liu17, Yue15,Dianwu96} for more information.

In \cite{WLP19}, let $m\geq 6$ be even, the authors described the largest coset leader $\delta_1$ modulo $n=\frac{q^m-1}{q+1}$ for $q=2,3$ and studied the parameters of the narrow-sense BCH codes with design distance $\delta_1$. In addition, the authors gave a conjecture on the largest coset leader  modulo $n=\frac{q^m-1}{q+1}$ for $q>3$. In \cite{GDL21}, the authors obtained a sufficient and necessary condition for the code $\mathcal{C}_{(q,\frac{q^m-1}{q-1},\delta)}$ being a dually-BCH code in the ternary case, while the case $q\geq 4$ is still open.

Below we always assume that $m\geq 4$ is an integer and $q$ is a prime power, $m$ is even if $n=\frac{q^m-1}{q+1}$ and $q>2$  if $n=\frac{q^m-1}{q-1}$.
The main objective
 of this paper is to give several sufficient and necessary conditions in terms of the designed distance to ensure that the BCH codes with length $n$ are dually-BCH codes
 and develop lower bounds on the minimum distances of the dual codes, where $n=\frac{q^m-1}{q+1}$.
This paper generalizes some results in \cite{GDL21} when $n=\frac{q^m-1}{q+1}$.
 As byproducts, we present the largest coset leader $\delta_1$ modulo $n$, which proves a conjecture in \cite{WLP19} and partially solves an open problem in \cite{Li2017}.
 We also investigate the parameters of the narrow-sense BCH codes of length $n$ with designed distance $\delta_1$, where $n=\frac{q^m-1}{q+1}$.
 To investigate the optimality
 of the codes studied in this paper, we compare them with the tables of the best known linear codes maintained in \cite{Grassl2006}, and
the best known cyclic codes  maintained in \cite{Dingbook15}, and some of  the proposed codes are optimal or almost optimal.

The rest of this paper is organized as follows. Section \ref{sec-auxiliary} contains some preliminaries. Sections \ref{sec-plus} and \ref{sec-minus} give the sufficient and necessary conditions in terms of the designed distance to ensure that the BCH codes with length $n$ are dually-BCH codes
 and develop lower bounds on the minimum distances of the dual codes for $n=\frac{q^m-1}{q+1}$. Section \ref{sec-finals} concludes the paper.

\section{Preliminaries}\label{sec-auxiliary}

In this section, we introduce some
basic concepts and known results on BCH codes, which will be used later.
Starting from now on, we adopt the following notation unless otherwise stated:
\begin{itemize}
\item $\gf(q)$ is the finite field with $q$ elements.
\item $\alpha$ is a primitive element of $\gf(q^m)$ and $\beta=\alpha^n$ is a primitive $n$-th root of unity, where $n \, |\, q^m-1$.
\item $m_i(x)$ denotes the minimal polynomial of $\beta^i$ over $\gf(q)$.
\item $g_{(q,n,\delta)}=\lcm \left(m_{1}(x), m_{2}(x), \cdots, m_{\delta-1}(x)\right)$ denotes the least common multiple of these minimal polynomials.
\item $\mathcal{C}_{(q,n,\delta)}$ denotes the BCH code with generator polynomial $g_{(q,n,\delta)}$.
\item $T=\{0\leq i\leq n-1: g(\beta^i)=0\}$ is the defining set of $\mathcal{C}_{(q,n,\delta)}$ with respect to $\beta$.
\item $T^{-1}=\{n-i\,:\,i\in T\}$.
\item $T^{\perp}$ is the defining set of the dual code $\mathcal{C}^{\perp}_{(q,n,\delta)}$ with respect to $\beta$.
\item $\lceil x \rceil$ denotes the smallest integer larger than or equal to $x$.
\item $\lfloor x \rfloor$ denotes the largest integer less than or equal to $x$.
\item $\overline{x}_t$ denotes $x \pmod{t}$ and $0\leq\overline{x}_t\leq t-1$, where $t$ is a positive integer.
\end{itemize}

 For two positive integers $0<a,\ b<q^m-1$, let $a=\sum\limits_{i=0}^{m-1}a_iq^i$ and $b=\sum\limits_{i=0}^{m-1}b_iq^i$ be the $q$-adic expansion of $a$ and $b$, respectively.
Write $\overline{a}=(a_{m-1},a_{m-2},\cdots,a_{0})_q$ and $\overline{b}=(b_{m-1},b_{m-2},\cdots,b_{0})_q$. Define $\overline{a}>\overline{b}$ if $a_{m-1}>b_{m-1}$ or there exists an integer $0\leq i\leq m-2$ such that $a_i> b_i$ and $a_j=b_j$ for $j\in [i+1,m-1]$. Then we have the following result.
\begin{lemma}\label{lem1b21}
Let $a,b$ be given as above. Then $a>b$ if and only if $\overline{a}>\overline{b}$.
\end{lemma}

Let $\mathbb{Z}_n$ denote the ring of integers modulo $n$.
Let $s$ be an integer with $0\leq s<n$. The {\it $q$-cyclotomic coset} of $s$ is defined by
$$C_s=\{s,sq,sq^2,\cdots,sq^{\ell_{s}-1}\}\,\, {\text\,\,{\rm mod} \,\,n\subseteq \mathbb{Z}_n, }$$
where $\ell_s$ is the smallest positive integer such that $s\equiv sq^{\ell_s} \pmod n$, and is the size of the $q$-cyclotomic coset.
 The smallest integer in $C_s$ is called the {\it coset leader} of $C_s$. The following lemma on coset leaders modulo $q^m-1$ will play an important role in proving the conjecture documented in \cite{WLP19}.
\begin{lemma}\label{lem:qm1q1}
The first three largest $q$-cyclotomic coset leaders modulo $n=q^m-1$ are given as follows:
$$\delta_1=(q-1)q^{m-1}-1,\,\,\delta_2= (q-1)q^{m-1}-q^{\lfloor\frac{m-1}{2}\rfloor}-1,\,\,\delta_3= (q-1)q^{m-1}-q^{\lfloor\frac{m+1}{2}\rfloor}-1.$$
\end{lemma}

The following lemma is given in \cite{GDL21}, which is useful  to give a characterization of $\mathcal{C}_{(q,\frac{q^m-1}{q-1},\delta)}$ being a dually-BCH code for $q\geq 3$.

\begin{lemma} \label{lemma-break-point-(q^m-1)/(q-1)}
	For $2 \le \delta < n$, let $I(\delta)$ be the integer such that $\{ 0,1,2,\ldots,I(\delta)-1 \} \subseteq T^{\bot}$ and $I(\delta) \notin T^{\bot}$.
	Then we have $I(\delta)=\frac{q^{m-t}-1}{q-1}$ if $\frac{q^t-1}{q-1}< \delta \le \frac{q^{t+1}-1}{q-1}\ (1 \le t \le m-2)$ and $I(\delta)=1$ if  $\frac{q^{m-1}-1}{q-1} < \delta < n$.
\end{lemma}


\section{BCH codes of length $n=\frac{q^m-1}{q+1}$ and its dual}\label{sec-plus}

Throughout this section, we always assume that $n=\frac{q^m-1}{q+1}$, where $m\geq 4$ is even. Recall that $\mathcal{C}_{(q,n,\delta)}$ is a BCH code with designed distance $\delta$, i.e., the defining set with respect to $\beta$ is $T=C_1\cup C_2 \cup \cdots \cup C_{\delta-1}$, where $2\leq \delta\leq n$ and $\beta=\alpha^{q+1}$. It then follows that $T^{\perp}=Z_n \setminus T^{-1}$ is the defining set of the dual code $\mathcal{C}^{\perp}_{(q,n,\delta)}$ with respect to $\beta$, where $T^{-1}=\{n-i\,:\,i\in T\}$. It is clear that $0 \in T^{\perp}$.
In this section, we  give the largest coset leader $\delta_1$  modulo $n$ and study the parameters of the BCH code $\mathcal{C}_{(q,n,\delta_1)}$.  We also derive a sufficient and necessary condition for $\mathcal{C}_{(q,n,\delta)}$ being a dually-BCH code and develop the lower bounds on the minimum distance of $\mathcal{C}^{\perp}_{(q,n,\delta)}$, where $2\leq \delta \leq n$.

\begin{lemma} \label{lemma-1219}
Let $h$ be a positive integer and let $q$ be a prime power such that $(q+1)\ |\ h$. Then $h$ is a coset leader modulo $q^m-1$ if and only if $\frac{h}{q+1}$ is a coset leader modulo $n$.
\end{lemma}

\begin{proof}
We first assume that $h$ is a coset leader modulo $q^m-1$. If $\frac{h}{q+1}$ is not a coset leader modulo $n$,
	then there would be an integer $\ell$ with $1 \le \ell \le m-1$ such that
\begin{equation*}
\frac{h}{q+1}q^\ell \bmod{n} <\frac{h}{q+1},
\end{equation*}
which implies that
\begin{equation*}
hq^\ell \bmod{q^m-1} <h.
\end{equation*}
This leads to a contradiction. Hence, $\frac{h}{q+1}$ is a coset leader modulo $n$.
With the same argument, we can see that $h$ is a coset leader modulo $q^m-1$ if $\frac{h}{q+1}$ is a coset leader modulo $n$.
	The proof is then completed.
\end{proof}

\begin{lemma}\label{detal1221}
Let $\delta_1$ be the largest coset leader  modulo $n$, then
\[\delta_1=\left\{ \begin{array}{lll}
           \frac{(q-1)q^{m-1}-q^{\frac{m-2}{2}}-1}{q+1}, & \, \,\, \text{if $m\equiv 2 \pmod 4$}, \\
\frac{(q-1)q^{m-1}-q^{\frac{m}{2}}-1}{q+1}, & \, \,\, \text{if $m\equiv 0 \pmod 4$}. \end{array}  \right.\]
\end{lemma}

\begin{proof}
We prove the desired conclusion only for the case  $m\equiv 0 \pmod 4$, and omit the proof for  the case $m\equiv 2 \pmod 4$, which is similar.

It is easy to check that
$(q+1)\,|\,(q^{m-1}+1)$ and $(q+1)\,|\, (q^{\frac{m}{2}-1}+1)$
by noting that $m\equiv 0 \pmod 4$.
Then
$$(q+1)\,|\,(q^{m-1}+1)+q^{\frac{m}{2}}(q^{\frac{m}{2}-1}+1)=2q^{m-1}+q^{\frac{m}{2}}+1,$$
which implies that $(q+1)\, | \, (q+1)q^{m-1}-(2q^{m-1}+q^{\frac{m}{2}}+1)=(q-1)q^{m-1}-q^{\frac{m}{2}}-1$.
Recall that $(q-1)q^{m-1}-q^{\frac{m}{2}}-1$ is a coset leader modulo $q^m-1$ from Lemma \ref{lem:qm1q1}.
It then follows from Lemma \ref{lemma-1219}  that $\frac{(q-1)q^{m-1}-q^{\frac{m}{2}}-1}{q+1}$ is a coset leader modulo $n$.

We now assert that $\frac{(q-1)q^{m-1}-q^{\frac{m}{2}}-1}{q+1}$ is the largest coset leader modulo $n$. If there exists a coset leader $\delta'$ modulo $n$ such that $\frac{(q-1)q^{m-1}-q^{\frac{m}{2}}-1}{q+1}<\delta'<\frac{q^m-1}{q+1}$, then for any positive integer $\ell$ with $1\leq \ell\leq m-1$, we have $\left(\delta' q^\ell \bmod{\frac{q^m-1}{q+1}}\right) \geq\delta'$, which implies that
	\begin{equation*}
 \left((q+1)\delta' q^\ell \bmod{q^m-1}\right) \geq(q+1)\delta'.
	\end{equation*}
Hence, $(q+1)\delta'$ is a coset leader modulo $q^m-1$. From Lemma \ref{lem:qm1q1},  we know that
\begin{equation}\label{eq:0710}
(q+1)\delta'=(q-1)q^{m-1}-1\,\,\text{or}\,\,(q+1)\delta'=(q-1)q^{m-1}-q^{\frac{m-2}{2}}-1.
\end{equation}
It is easy to check that Equation \eqref{eq:0710} is impossible since $\delta'$ is an integer.
The desired conclusion then follows.
\end{proof}

  A conjecture on the largest coset leader  modulo $n$ was given in \cite{WLP19} for $q>3$. In Lemma \ref{detal1221}, we answered this problem.
  We are now ready to determine $|C_{\delta_1}|$ that is useful to determine the dimension of the BCH code $\mathcal{C}_{(q,n,\delta_1)}$.

\begin{lemma} \label{lem-delta-plus}
Let $\delta_1$ be given as in Lemma \ref{detal1221}, then $|C_{\delta_1}|=m$ if $m\equiv 0 \pmod 4$ and $|C_{\delta_1}|=\frac{m}{2}$ if $m\equiv 2 \pmod 4$.
\end{lemma}
\begin{proof}
We prove the desired conclusion only for the case $m\equiv 0 \pmod 4$, and omit the proof for the case $m\equiv 2 \pmod 4$, which is similar.

Let $|C_{\delta}|=\ell$, then
$$\frac{q^m-1}{q+1}\, |\, \left(\frac{(q-1)q^{m-1}-q^{\frac{m}{2}}-1}{q+1}\right)(q^{\ell}-1),$$
which is equivalent to
$$q^m-1\, |\, \left((q-1)q^{m-1}-q^{\frac{m}{2}}-1\right)(q^{\ell}-1).$$
It then follows from Lemma \ref{lem:qm1q1} that $\ell=m$.
\end{proof}

The following theorem gives the information on the parameters of  the BCH code $\mathcal{C}_{(q,n,\delta_1)}$, where $n=\frac{q^m-1}{q+1}$.

\begin{theorem}\label{theorem1}
 When $m\equiv 2 \pmod 4$, then the BCH code $\mathcal{C}_{(q,n,\delta_1)}$ has parameters
$$\left[\frac{q^m-1}{q+1}, \frac{m}{2}+1, d\geq \frac{(q-1)q^{m-1}-q^{\frac{m-2}{2}}-1}{q+1}\right].$$
 When $m\equiv 0 \pmod 4$, then the BCH code $\mathcal{C}_{(q,n,\delta_1)}$ has parameters
$$\left[\frac{q^m-1}{q+1}, m+1, d\geq \frac{(q-1)q^{m-1}-q^{\frac{m}{2}}-1}{q+1}\right].$$
\end{theorem}

\begin{proof}
It is straightforward from Lemmas \ref{detal1221} and \ref{lem-delta-plus}.
\end{proof}

We inform the reader that Theorem \ref{theorem1} generalizes the results in \cite[Theorems 14 and 17]{WLP19}, where the only the binary and ternary cases  were involved. In this sense, the results documented in
\cite{WLP19} can be seen as a special case of Theorem \ref{theorem1}. The following example shows that the lower bounds given in Theorem \ref{theorem1} are very good.

\begin{example}\label{example-02}
We have the following examples for the code of Theorem \ref{theorem1}.
\begin{itemize}
\item Let $q=2$, $m=6$ and $\delta=9$, then the code $\mathcal{C}_{(2,21,9)}$ has parameters $[21,4,\geq 9]$.
\item Let $q=3$, $m=4$ and $ \delta =11$, then the code $\mathcal{C}_{(3,20,11)}$ has parameters $[20,5,\geq 11]$.
\item Let $q=4$, $m=4$ and $ \delta =35$, then the code $\mathcal{C}_{(4,51,35)}$ has parameters $[51,5,\geq 35]$.
\item Let $q=5$, $m=4$ and $\delta=79$, then the code $\mathcal{C}_{(5,104,79)}$ has parameters $[104,5,\geq 79]$.
\end{itemize}
All the four codes are almost optimal according to the tables of best codes known  in \cite{Grassl2006} when the equality holds.
 These results are verified by Magma programs.
\end{example}

To present a sufficient and necessary condition for $\mathcal{C}_{(q,n,\delta)}$ being a dually-BCH code,
 the following several lemmas will be needed later.

\begin{lemma}\label{lemma:1222}
Let $t\geq 2$ be even, $m\geq 4$ even and $l$  odd. Then $\frac{q^l+1}{q+1}$, $\frac{q^t-1}{q+1}$, $\frac{\sum_{i=0}^{m-1}q^i}{q+1}$, $\frac{(q-2)(\sum_{i=0}^{m-1}q^i)}{q+1}$ and $\frac{q^m-q^{m-1}-q^{m-2}-1}{q+1}$ are coset leaders modulo $n$.
\end{lemma}

\begin{proof}
We only prove the desired conclusion for the case that $\frac{q^m-q^{m-1}-q^{m-2}-1}{q+1}$ is a coset leader modulo $n$, and omit the proofs of
 other cases, which are similar.

 It is easy to check that $(q+1) \  \mid \ (q^m-q^{m-1}-q^{m-2}-1)$, and
 $$\overline{q^m-q^{m-1}-q^{m-2}-1}=(q-2,q-2,q-1,q-1,\cdots,q-1)_q.$$
By Lemma \ref{lem1b21}, we see that $q^m-q^{m-1}-q^{m-2}-1$ is a coset leader modulo $q^m-1$.  The desired conclusion then follows from Lemma \ref{lemma-1219}.
\end{proof}

\begin{lemma}\label{lem:122501}
Let $\delta_1$ be given as in Lemma \ref{detal1221}. Then we have
\begin{enumerate}
\item[{\rm (1)}]  $\delta_1 \in T^\bot$ is a coset leader modulo $n$ if $m\equiv 0 \pmod 4$, $m\geq 8$ and $2 \le \delta \le\frac{q^{\frac{m}{2}-1}+1}{q+1}$.
\item[{\rm (2)}]  $\delta_1 \in T^\bot$ is a coset leader modulo $n$ if $m\equiv 2 \pmod 4$, $m\geq 6$ and $2 \le \delta \le \frac{q^{\frac{m}{2}}+1}{q+1}$.
\end{enumerate}
\end{lemma}

\begin{proof}
For the first case, recall from Lemma \ref{lemma:1222} that $\frac{q^{\frac{m}{2}-1}+1}{q+1}$ is a coset leader modulo $n$.  This means that $\frac{q^{\frac{m}{2}-1}+1}{q+1}  \nsubseteq T$
if $2 \le \delta \le \frac{q^{\frac{m}{2}-1}+1}{q+1}$. Hence,
$\delta_1=n-\frac{q^{\frac{m}{2}}(q^{\frac{m}{2}-1}+1)}{q+1}\notin T^{-1}$, i.e., $\delta_1 \in Z_n \backslash T^{-1}=T^{\perp}$.
It is similar to give the proof for the second case, and we omit the details.
\end{proof}

\begin{lemma}\label{lem:10}
Let $q$ be an odd prime power. Then $\frac{(q-1)q^3-q^2-q-2}{q+1}$ is the second largest coset leader modulo $\frac{q^4-1}{q+1}$.
\end{lemma}

\begin{proof}
 It is easy to see  that $$\overline{(q-1)q^3-q^2-q-2}=(q-2,q-2,q-2,q-2)_q,$$ so $(q-1)q^3-q^2-q-2$ is a coset leader modulo $q^4-1$.
From Lemma \ref{lemma-1219}, we know that $\frac{(q-1)q^3-q^2-q-2}{q+1}$ is a coset leader modulo $\frac{q^4-1}{q+1}$.

 We now assert that $\frac{(q-1)q^3-q^2-q-2}{q+1}$ is the second largest coset leader modulo $\frac{q^4-1}{q+1}$. If there exists a coset leader $\delta'$ modulo $n$ such that $\frac{(q-1)q^3-q^2-q-2}{q+1}<\delta'<\frac{(q-1)q^3-1}{q+1}$, then for any positive integer $\ell$ with $1\leq \ell\leq m-1$, we have $\left(\delta' q^\ell \bmod{\frac{q^m-1}{q+1}}\right) \geq\delta'$, which implies that
	\begin{equation*}
 \left((q+1)\delta' q^\ell \bmod{q^m-1}\right) \geq(q+1)\delta'.
	\end{equation*}
Hence, $(q+1)\delta'$ is a coset leader modulo $q^m-1$.

From Lemma \ref{lem1b21}, we know that the sequences of first five largest coset leaders modulo $q^4-1$ are
$(q-2,q-1,q-1,q-1)_q,$ $(q-2,q-1,q-2,q-1)_q,$ $(q-2,q-2,q-1,q-1)_q,$
$(q-2,q-2,q-2,q-1)_q$ and $(q-2,q-2,q-2,q-2)_q$, respectively.
Then the first five largest coset leaders modulo $q^4-1$ are $(q-1)q^3-1$, $(q-1)q^3-q-1$, $(q-1)q^3-q^2-1$, $(q-1)q^3-q^2-q-1$ and $(q-1)q^3-q^2-q-2$, respectively.
Since $\frac{(q-1)q^3-q^2-q-2}{q+1}<\delta'<\frac{(q-1)q^3-1}{q+1}$,  if $\frac{(q-1)q^3-q^2-q-2}{q+1}$ is not the second largest coset leader, we have
\begin{equation}\label{eq:0729}
(q+1)\delta'=(q-1)q^3-q-1, \text{or}\,\,(q+1)\delta'=(q-1)q^3-q^2-1,\text{or}\,\,(q+1)\delta'=(q-1)q^3-q^2-q-1.
\end{equation}
It is easy to check that Equation \eqref{eq:0729} is impossible since $\delta'$ is an integer.
The desired conclusion then follows.
\end{proof}

\begin{lemma}\label{lem:122201}
Let $q>2$ be a prime power and $m\geq 4$ be even. Let $\delta_1$ be given as in Lemma \ref{detal1221}. Then the following hold.
\begin{enumerate}
\item[{\rm (1)}]  $\frac{(q-1)q^3-q^2-q-2}{q+1} \in T^\bot$ is a coset leader modulo $n$ if $m= 4$ and $2 \le \delta \le q^2+1$.
\item[{\rm (2)}]  $\frac{\sum_{i=0}^{m-1}q^i}{q+1} \in T^\bot$ is a coset leader modulo $n$ if $2 \le \delta \le \frac{(q-2)(\sum_{i=0}^{m-1}q^i)}{q+1}$.
\item[{\rm (3)}]  $2 \notin T^\bot$ is a coset leader modulo $n$ if $\frac{(q-2)(\sum_{i=0}^{m-1}q^i)}{q+1}  < \delta \le \delta_1$.
\end{enumerate}
\end{lemma}

\begin{proof} From Lemma \ref{lemma-1219} and Lemma \ref{lem:10}, it is easy to get that $\frac{(q-1)q^3-q^2-q-2}{q+1}$, $\frac{\sum_{i=0}^{m-1}q^i}{q+1}$, and $2$ are coset leaders modulo $n$.

When $m=4$ and $2 \le \delta \le q^2+1$, it follows from Lemma \ref{lemma:1222} that
	$\frac{q^3+q^2+q+1}{q+1}=q^2+1$ is a coset leader modulo $\frac{q^4-1}{q+1}$. This means that $q^2+1  \nsubseteq T$
in this case. Hence,
$$\frac{(q-1)q^3-q^2-q-2}{q+1}=n-(q^2+1)\notin T^{-1} \text{ and } \frac{(q-1)q^3-q^2-q-2}{q+1} \in Z_n \backslash T^{-1}=T^{\perp}.$$

When $2 \le \delta \le \frac{(q-2)(\sum_{i=0}^{m-1}q^i)}{q+1}$,
 we see from Lemma \ref{lemma:1222} that
	$\frac{(q-2)(\sum_{i=0}^{m-1}q^i)}{q+1}$ is a coset leader modulo $n$.  Then $\frac{(q-2)(\sum_{i=0}^{m-1}q^i)}{q+1} \notin T$. Hence,
$$\frac{\sum_{i=0}^{m-1}q^i}{q+1}=n-\frac{(q-2)(\sum_{i=0}^{m-1}q^i)}{q+1}\notin T^{-1} \text{ and } \frac{\sum_{i=0}^{m-1}q^i}{q+1} \in Z_n \backslash T^{-1}=T^{\perp}.$$

When $\frac{(q-2)(\sum_{i=0}^{m-1}q^i)}{q+1}  < \delta \le \delta_1$, it is easy to see that
$$\frac{(q-2)q^{m-1}-2q^{m-2}-1}{q+1}<\frac{(q-2)(\sum_{i=0}^{m-1}q^i)}{q+1}.$$
Then $$\frac{(q-2)q^{m-1}-2q^{m-2}-1}{q+1}=\frac{q^m-1}{q+1}-2q^{m-2} \in T.$$
Note that $n-2q^{m-2}$ and $n-2$ are in the same $q$-cyclotomic coset modulo $n$. Then we have
 $n-2\in T$, i.e.,
$2\in T^{-1}$. Hence, $2\notin T^{\perp}= Z_n\setminus T^{-1}.$
\end{proof}

\begin{lemma}\label{lem:0917}
Let $q=2$ and $m\geq 8$ be even. When $7< \delta \leq \frac{2^{m-2}-1}{3}$,  there is a coset leader belonging to $T^{\perp}$. Moreover, the coset leader is larger than $\frac{\sum_{i=0}^{m-5}2^i}{3}$.
\end{lemma}
\begin{proof}
Let $M=2^8+2^6+2^5+2^4+2^2+2+1=(0,1,0,1,1,1,0,1,1,1)_2$ if $m=10$, and
\begin{equation*}
\begin{split}
M&= 2^{m-2}+\sum_{i=0}^1(2^{m-4-4i}+2^{m-5-4i}+2^{m-6-4i})+\sum_{i=0}^{\frac{m-13}{3}}(2^{m-12-3i}+2^{m-13-3i})\\
&=(0,1,0,1,1,1,0,1,1,1,\underline{0,1,1},\underline{0,1,1},\cdots,\underline{0,1,1})_2
\end{split}
\end{equation*}
if $m\equiv 1 \pmod 3$ and $m\neq 10$, and
\[M=\left\{ \begin{array}{lll}
           \sum_{i=0}^{\frac{m-2}{2}}2^{2i}=(0,1,0,1,\cdots,0,1)_2, & \, \,\, \text{if $m\equiv 0 \pmod 3$},  \\
\sum_{i=1}^{\frac{m-2}{2}}2^{2i}+2+1=(\underline{0,1},\underline{0,1},\cdots,\underline{0,1},1,1)_2, & \, \,\, \text{if $m\equiv 2 \pmod 3$}. \end{array}  \right.\]
It is easy to see that $M$ is a coset leader modulo $2^m-1$. By the definition of $M$, we know that $3\, |\, M$. Then from Lemma \ref{lemma-1219}, $\frac{M}{3}$ is a coset leader modulo $n$.

Since $\frac{M}{3}>\frac{2^{m-2}-1}{3}$, we have $M \notin T$. If $m=10$, then
$$\frac{2^3(2^6+2^4+1)}{3}=n-M \notin T^{-1} \text{ and } \frac{2^6+2^2+1}{3} \in Z_n \backslash T^{-1}=T^{\perp}$$
since $\frac{2^6+2^2+1}{3}$ and $\frac{2^6+2^4+1}{3}$ are in the same coset modulo $n$.
If $m\equiv 1 \pmod 3$ and $m\neq 10$, then
$\frac{2^2(2^{m-3}+2^{m-5}+2^{m-9}+\sum_{i=0}^{\frac{m-13}{3}}2^{3i})}{3}=n-M \notin T^{-1}$ and  $$\frac{2^{m-4}+\sum_{i=0}^{\frac{m-13}{3}}2^{3i+5}+2^2+1}{3} \in Z_n \backslash T^{-1}=T^{\perp}$$
since $\frac{2^{m-3}+2^{m-5}+2^{m-9}+\sum_{i=0}^{\frac{m-13}{3}}2^{3i}}{3}$ and $\frac{2^{m-4}+\sum_{i=0}^{\frac{m-13}{3}}2^{3i+5}+2^2+1}{3}$ are in the same coset modulo $n$.

Similarly, we have
\[\left\{ \begin{array}{lll}
           \frac{\sum_{i=0}^{\frac{m-2}{2}}2^{2i}}{3}\in Z_n \backslash T^{-1}=T^{\perp}, & \, \,\, \text{if $m\equiv 0 \pmod 3$},  \\
\frac{\sum_{i=0}^{\frac{m-4}{2}}2^{2i}}{3}\in Z_n \backslash T^{-1}=T^{\perp}, & \, \,\, \text{if $m\equiv 2 \pmod 3$}. \end{array}  \right.\]
Let
\[L=\left\{ \begin{array}{lll}
          2^6+2^2+1, & \, \,\, \text{if $m=10$},  \\
           \sum_{i=0}^{\frac{m-2}{2}}2^{2i}, & \, \,\, \text{if $m\equiv 0 \pmod 3$},  \\
            2^{m-4}+\sum_{i=0}^{\frac{m-13}{3}}2^{3i}+2^2+1, & \, \,\, \text{if $m\equiv 1 \pmod 3$ and $m \ne 10$},  \\
\sum_{i=0}^{\frac{m-4}{2}}2^{2i}, & \, \,\, \text{if $m\equiv 2 \pmod 3$}. \end{array}  \right.\]
It is easy to check that $L$ is a coset leader modulo $2^m-1$.  From Lemma \ref{lemma-1219}, $\frac{L}{3}$ is a coset leader modulo $n$. It is obvious that $\frac{L}{3}>\frac{\sum_{i=0}^{m-5}2^i}{3}$. The desired conclusion then follows.
\end{proof}

\begin{lemma} \label{lemma-break-point-(q^m-1)/(q-1)}
Let $q$ be a prime power  and $2\leq t \leq m-2$ be even. For $2 \le \delta < n$, let $I(\delta) \ge 1$ be the integer such that $\{ 0,1,2,\ldots,I(\delta)-1 \} \subseteq T^{\bot}$ and $I(\delta) \notin T^{\bot}$.
Then we have $I(\delta)=\frac{q^{m-t}-1}{q+1}$ if $\frac{q^{t}-1}{q+1}< \delta \le \frac{q^{t+1}+2q^t-1}{q+1}$.
\end{lemma}

\begin{proof}
When $\frac{q^{t}-1}{q+1}< \delta \le \frac{q^{t+1}+2q^t-1}{q+1}$, it is straightforward to see that
$$\frac{q^m-q^{m-t}}{q+1}= \frac{(q^{t}-1)}{q+1}  q^{m-t} \in C_{\frac{q^{t}-1}{q+1}} \subseteq T.$$
Therefore, $\frac{q^{m-t}-1}{q+1}=n-\frac{q^m-q^{m-t}}{q+1} \in T^{-1}$ and $\frac{q^{m-t}-1}{q+1} \not \in T^\bot=\Bbb Z_n \setminus T^{-1}$.
	
	We are ready to show that $\{ 0,1,2,\ldots,\frac{q^{m-t}-1}{q+1}-1 \} \subseteq T^{\bot}$. It is clear that $0 \in T^\bot$.
	For every integer $i$ with $1 \le i \le \frac{q^{m-t}-1}{q+1}-1$, we have $i=\frac{q^{m-t}-1}{q+1}-u$, where $1 \le u \le \frac{q^{m-t}-1}{q+1}-1$. Note that
	$$((q+1) \cdot q^{t}u+q^{t}-1)q^{m-t} \equiv q^m-q^{m-t}+(q+1)u \pmod {q^m-1}$$
	and the sequence of $(q+1) \cdot q^{t}u+q^{t}-1$ is
	$$(\underbrace{i_{m-1},i_{m-2},\ldots,i_{t}}_{m-t}, \underbrace{q-1,\ldots,q-1}_{t})_q,$$
	where $i_{t}=i_{t+1}=1$ if $u=1$ and there are at least two $t+1 \le j \le m-1$ such that $i_j \ne 0$ if $u >1$. It follows that
	the coset leader of the cyclotomic coset of $(q+1) \cdot q^{t}u+q^{t}-1$ modulo $q^m-1$ is larger than or equal to $q^{t+1}+2q^t-1$.
	Then we obtain that
	$$\text{CL}\left(\frac{(q+1) \cdot q^{t}u+q^{t}-1}{q+1}\right) \ge \frac{q^{t+1}+2q^t-1}{q+1} > \delta-1,$$
where $\text{CL}\left(\frac{(q+1) \cdot q^{t}u+q^{t}-1}{q+1}\right)$ denotes the coset leader of the $2$-cyclotomic coset modulo $n$ containing $\frac{(q+1) \cdot q^{t}u+q^{t}-1}{q+1}$.
Consequently, $\frac{(q+1) \cdot q^{t}u+q^{t}-1}{q+1} \not \in T$ and $\frac{q^{m-t}-1}{q+1}-u \not \in T^{-1}$. This leads to $i=\frac{q^{m-t}-1}{q+1}-u \in T^\bot$.
	It then follows that $I(\delta)=\frac{q^{m-t}-1}{q+1}$ for any $\delta$ with $\frac{q^{t}-1}{q+1}< \delta \le \frac{q^{t+1}+2q^t-1}{q+1}$.
 The proof is then completed.
\end{proof}

\begin{lemma} \label{lemma-1002}
Let $q$ be a prime power  and $2\leq \delta \leq q-1$. Let $d^{\perp}(\delta)$ be the minimum distance of $\mathcal C_{(q,n,\delta)}^\bot$.  Then
$d^{\perp}(\delta)\geq \frac{q^{m-1}+2q^{m-2}-1}{q+1}$.
\end{lemma}

\begin{proof}
From the BCH bound, in order the obtain the desired result, we only need to show that $\{ 0,1,2,\ldots,\frac{q^{m-1}+2q^{m-2}-1}{q+1}-1 \} \subseteq T^{\bot}$.

 It is clear that $0 \in T^\bot$.
For every integer $i$ with $1 \le i \le \frac{q^{m-1}+2q^{m-2}-1}{q+1}-1$, we have $i=\frac{q^{m-1}+2q^{m-2}-1}{q+1}-u$, where $1 \le u \le \frac{q^{m-1}+2q^{m-2}-1}{q+1}-1$.
Since $q+1\,|\,q^{m-1}+2q^{m-2}-1$ and $q+1\,|\,q^m-1$, if there exist $0<t \leq q-1$ and $0\leq i \leq m-1$ such that
$$ q^{m-1}+2q^{m-2}-1-u(q+1) \equiv tq^i \pmod {q^m-1},$$
then
$q+1\,|\,tq^i$, which is impossible. Hence,
	$$\text{CL}\left(\frac{q^{m-1}+2q^{m-2}-1-u(q+1)}{q+1}\right) > q-1>\delta,$$
where $\text{CL}\left(\frac{q^{m-1}+2q^{m-2}-1-u(q+1)}{q+1}\right)$ denotes the coset leader of the $2$-cyclotomic coset modulo $n$ containing $\frac{q^{m-1}+2q^{m-2}-1-u(q+1)}{q+1}$.
Consequently, $\frac{q^{m-1}+2q^{m-2}-1-u(q+1)}{q+1} \not \in T$ and $\frac{q^{m-1}+2q^{m-2}-1}{q+1}-u \not \in T^{-1}$. This leads to $i=\frac{q^{m-1}+2q^{m-2}-1}{q+1}-u \in T^\bot$.
From the BCH bound, the desired conclusion then follows.
\end{proof}

Let $2\leq \delta'\leq\delta''\leq n$, it is clear that $\mathcal{C}_{(q,n,\delta'')}\subseteq \mathcal{C}_{(q,n,\delta')}$. Then we have $\mathcal{C}_{(q,n,\delta'')}^{\perp}\supseteq \mathcal{C}_{(q,n,\delta')}^{\perp}$, which implies that $d(\mathcal{C}_{(q,n,\delta'')}^{\perp})\leq d(\mathcal{C}_{(q,n,\delta')}^{\perp}).$
From Lemma \ref{lemma-break-point-(q^m-1)/(q-1)}, Lemma \ref{lemma-1002} and the BCH bound for cyclic codes, it is easy to get the minimum distance of the lower bound of $\mathcal C_{(q,n,\delta)}^\bot$.

\begin{theorem}\label{thm:1002}
Let $2\leq \delta \leq n$, $0\leq t\leq m-2$ and $d^{\perp}(\delta)$ be the minimum distance of $\mathcal C_{(q,n,\delta)}^\bot$. Let $q=2$, then we have $d^{\perp}(\delta)\geq\frac{2^{m-t}-1}{3}+1$ if $\frac{2^{t}-1}{3}< \delta \le \frac{2^{t+2}-1}{3}$. Let $q>2$ be a prime power, then we have
\[d^{\perp}(\delta)\geq\left\{ \begin{array}{lll}
           \frac{q^{m-1}+2q^{m-2}-1}{q+1}, & \, \,\, \text{if $2\leq \delta \leq q-1$}, \\
\frac{q^{m-t}-1}{q+1}+1, & \, \,\, \text{if $\frac{q^{t}-1}{q+1}< \delta \le \frac{q^{t+1}+2q^t-1}{q+1}$},\\
\frac{q^{m-t-2}-1}{q+1}+1\, (t\neq m-2), & \, \,\, \text{if $ \frac{q^{t+1}+2q^t-1}{q+1}< \delta \le \frac{q^{t+2}-1}{q+1}$},\\
2, & \, \,\, \text{if $\frac{q^{m-1}+2q^{m-2}-1}{q+1}< \delta \le \frac{q^m-1}{q+1}$}.\\
 \end{array}  \right.\]
\end{theorem}

It is very hard to determine the minimum distance of $\mathcal C_{(q,n,\delta)}^\bot$ in general. The following examples show that the lower bounds in Theorem \ref{thm:1002} are good in some cases.

\begin{example}
Let $\delta=2$, $q=3$ and $m=4$. In theorem \ref{thm:1002}, the minimum distance of the lower bound of $\mathcal C_{(3,20,2)}^\bot$ is $11$. By Magma, the true minimum distance of $\mathcal C_{(3,20,2)}^\bot$ is $12$.
\end{example}

\begin{example}
Let $\delta=2$, $q=2$ and $m=6$. In theorem \ref{thm:1002}, the minimum distance of the lower bound of $\mathcal C_{(2,21,2)}^\bot$ is $6$.  By Magma, the true minimum distance of $\mathcal C_{(3,21,2)}^\bot$ is $8$.
\end{example}

We now give the sufficient and necessary condition for $\mathcal{C}_{(q,n,\delta)}$ being a dually-BCH code. The cases $q=2$ and $q\neq 2$ will be treated separately.

\begin{theorem}\label{theorem2}
Let $n=\frac{q^m-1}{q+1}$, where $q = 2$ and $m \ge 4$ is even. Then $\mathcal C_{(q,n,\delta)}$ is a dually-BCH code if and only if
 $\delta_1+1 \le \delta \le n$, where $\delta_1$ is given in Lemma \ref{detal1221}.
\end{theorem}

\begin{proof}
We only prove the desired conclusion for the case $m\equiv 0 \pmod 4$, and omit the proof of
the case $m\equiv 2 \pmod 4$, which is similar.

By definition, we have $0 \notin T$ and $1 \in T $, then $0 \notin T^{-1}$ and  $n-1 \in T^{-1}$. Furthermore, we have $0 \in T^{\bot}$,
which means that $C_0$ must be the initial cyclotomic coset of $T^\bot$. In other words, there must be an integer $J \ge 1$
such that $T^\bot=C_0 \cup C_1 \cup \cdots \cup C_{J-1}$ if $\mathcal C_{(q,n,\delta)}$ is a dually-BCH code.
		
When $\delta_1+1 \le \delta \le n$,
it is easily seen that $T^\bot = \{0\}$ and $\mathcal C_{(q,n,\delta)}^\bot$ is a BCH code with respect to $\beta$ since $\delta_1$ is the largest coset leader modulo $n$.
		
It remains to prove the desired conclusion for $2 \le \delta \leq \delta_1$. If $m=4$, there is nothing to prove since $\delta_1=1$.
When $m>4$, we have the following three cases.

\noindent {\bf Case 1:} $2\leq \delta \leq 3$. From Lemma \ref{lem:122501} and $0 \in T^{\perp}$, we know that $T^{\perp}=C_0\bigcup C_1\bigcup\cdots \bigcup C_{\delta_1}=Z_n$ if $\mathcal{C}^{\perp}_{(2,n,\delta)}$ is a BCH code, which leads to $\mathcal{C}_{(2,n,\delta)}=\{\mathbf{0}\}$. It is obvious that $\mathcal{C}_{(2,n,\delta)}\neq\{\mathbf{0}\}$, which is a contradiction.

\noindent {\bf Case 2:} $3< \delta \leq 7$. It is clear that there exists $2\leq t\leq m-4$ such that $\frac{2^{t}-1}{3}< \delta \le \frac{2^{t+2}-1}{3}$ if $3< \delta \leq7$. It is easily seen that
$$2^{m-2}+2^{m-4}+ \sum_{i=0}^{m-6}2^i=(0,1,0,1,0, \underbrace{1,\ldots,1}_{m-5})_2$$
is a coset leader modulo $2^m-1$. From Lemma \ref{lemma-1219}, we know that $\frac{2^{m-2}+2^{m-4}+ \sum_{i=0}^{m-6}2^i}{3}$ is a coset leader modulo $\frac{2^m-1}{3}$. Obviously, $21=(0, \ldots, 0,1,0,1,0,1)_2$ is a coset leader modulo $2^m-1$, then $7$ is a  coset leader modulo $\frac{2^m-1}{3}$ form Lemma \ref{lemma-1219}.  Since $7\cdot2^{m-5}  \notin T$, we have
$$\frac{2^{m-2}+2^{m-4}+ \sum_{i=0}^{m-6}2^i}{3}=n-7\cdot2^{m-5} \notin T^{-1} \text{ and } \frac{2^{m-2}+2^{m-4}+ \sum_{i=0}^{m-6}2^i}{3} \in Z_n \backslash T^{-1}=T^{\perp}.$$
It is easy to check that
$$I_{\max}:=\max\left\{I(\delta): 3< \delta \leq \frac{2^{m-2}-1}{3}\right\}=I(4)=\frac{2^{m-2}-1}{3}<\frac{2^{m-2}+2^{m-4}+ \sum_{i=0}^{m-6}2^i}{3}.$$
It then follows that there is no integer $J \ge 1$ such that $T^\bot=C_0 \cup C_1 \cup \cdots \cup C_{J-1}$, i.e., $\mathcal C_{(q,n,\delta)}^\bot$ is not a BCH code with respect to $\beta$.

\noindent {\bf Case 3:} $7< \delta \leq \frac{2^{m-2}-1}{3}$. It is clear that $5=\frac{2^3+2^2+2+1}{3} \in T$, then
$$\frac{2^4(\sum_{i=0}^{m-5}2^i)}{3}=n-5 \in T^{-1} \text{ and } \frac{\sum_{i=0}^{m-5}2^i}{3} \in Z_n \backslash T^{-1}\notin T^{\perp}.$$
From Lemma \ref{lem:0917}, there is no integer $J \ge 1$ such that $T^\bot=C_0 \cup C_1 \cup \cdots \cup C_{J-1}$, i.e., $\mathcal C_{(q,n,\delta)}^\bot$ is not a BCH code with respect to $\beta$.

\noindent {\bf Case 4:} $\frac{2^{m-2}-1}{3}< \delta \le \delta_1$. Since $1\notin T^{\perp}$,  we have $\mathcal{C}^{\perp}_{(2,n,\delta)}=\{\mathbf{0}\}$ if $\mathcal{C}^{\perp}_{(2,n,\delta)}$ is a BCH code with respect to $\beta$. However, the dimension of $\mathcal{C}_{(2,n,\delta)}$ is ${\rm dim}(\mathcal{C}_{(2,n,\delta)})\leq n-|C_{\delta_1}|<n$, which is contradictory to ${\rm dim}(\mathcal{C}_{(2,n,\delta)})+{\rm dim}(\mathcal{C}^{\perp}_{(2,n,\delta)})=n$.

Combining all the cases above, the desired conclusion then follows.
\end{proof}

\begin{theorem}\label{theorem3}
Let $n=\frac{q^m-1}{q+1}$, where $q > 2$ is a prime power and $m \ge 4$ is even. Let $\delta_1$ be given in Lemma \ref{detal1221}. Then the following statements hold.
\begin{enumerate}
\item[{\rm (1)}] If $m=4$, then $\mathcal C_{(q,n,\delta)}$ is a dually-BCH code if and only if
		\begin{center}
			$\delta=2$, $\delta_1 \le \delta \le n$.
		\end{center}
\item[{\rm (2)}] If $m\neq4$, then $\mathcal C_{(q,n,\delta)}$ is a dually-BCH code if and only if
		\begin{center}
			$\delta_1+1 \le \delta \le n$.
		\end{center}
\end{enumerate}
	\end{theorem}
	
	\begin{proof}
We only prove the conclusion of this lemma for the case that $m\equiv 0 \pmod 4$, and omit the proof of
the conclusion for $m\equiv 2 \pmod 4$, which is similar.

With an analysis similar as Theorem \ref{theorem2}, when $\delta_1+1 \le \delta \le n$,
 we know that $T^\bot = \{0\}$ and $\mathcal C_{(q,n,\delta)}^\bot$ is a BCH code with respect to $\beta$.
It remains to show that whether $\mathcal C_{(q,n,\delta)}^\bot$ is a BCH code with respect to $\beta$ for $2 \le \delta \leq \delta_1$. We have the following four cases.

\noindent {\bf Case 1:} $\delta=2$. The defining set of $\mathcal C_{(q,n,\delta)}$ with respect to $\beta$ is $T=C_1$. Since $q^{m-2} \in T$, we have
$$n-q^{m-2}=\frac{q^m-q^{m-1}-q^{m-2}-1}{q+1}\in T^{-1}.$$
Let $\delta'=\frac{q^m-q^{m-1}-q^{m-2}-1}{q+1}.$ From Lemma \ref{lemma:1222}, we know that $\delta'$ is a coset leader modulo $n$. Hence, we obtain $T^{-1}=C_{\delta'}$.

If $m=4$, then $\delta'=\delta_1$, i.e., $T^{-1}=C_{\delta_1}$. Hence, $T^{\perp}=Z_n\backslash T^{-1}=C_0\cup C_1\cup\cdots C_{\delta_1-1}.$
This means that $\mathcal{C}^{\perp}_{(2,n,\delta)}=\mathcal{C}_{(2,n,\delta_1+1,0)}$ is a BCH code with the designed distance $\delta_1+1$ with respect to $\beta$.

If $m\neq4$,  then $\delta'<\delta_1$.
 Since $\delta'$ and $\delta_1$ are not in the same coset, then $\delta_1 \notin T^{-1}$, i.e., $\delta_1 \in Z_n \setminus T^{-1}=T^{\perp}$. It then follows $T^{\perp}=C_0\cup C_1\cup\cdots \cup C_{\delta_1}=Z_n$ if $\mathcal{C}^{\perp}_{(2,n,\delta)}$ is a BCH code. This means that $\mathcal{C}_{(2,n,\delta)}=\{\mathbf{0}\}$ and leads to a contradiction.

\noindent {\bf Case 2:} $3\le \delta \le q+1$. If $m=4$,  from Lemma \ref{lem:122201} and $0 \in T^{\perp}$, we know that $$T^{\perp}\supseteq C_0 \cup C_1\cup\cdots \cup C_{\frac{(q-1)q^3-q^2-q-2}{q+1}}$$ if  $\mathcal{C}^{\perp}_{(2,n,\delta)}$ is a BCH code.
From Lemma \ref{lem:10}, we know that the dimension of $\mathcal{C}^{\perp}_{(2,n,\delta)}$ is
\begin{equation}\label{eq:dsdf}
{\rm dim}(\mathcal{C}^{\perp}_{(2,n,\delta)})\geq \frac{q^4-1}{q+1}-|C_{\delta_1}|-|C_0|= \frac{q^4-1}{q+1}-5
\end{equation}
since $|C_{\delta_1}|=4$.
It is easy to check that $|C_1|=|C_2|=4$, then ${\rm dim}(\mathcal{C}_{(2,n,\delta)})\geq |C_2|+|C_3|\geq 8$. From  (\ref{eq:dsdf}), we know that ${\rm dim}(\mathcal{C}_{(2,n,\delta)})=\frac{q^4-1}{q+1}-{\rm dim}(\mathcal{C}^{\perp}_{(q,n,\delta)})\leq 5$ and this leads to a contradiction.
Hence,  $\mathcal{C}^{\perp}_{(q,n,\delta)}$ is not a BCH code.

 If $m\neq4$, from Lemma \ref{lem:122501} and $0 \in T^{\perp}$, we know that $T^{\perp}=C_0\bigcup C_1\bigcup\cdots \bigcup C_{\delta_1}=Z_n$ if $\mathcal{C}^{\perp}_{(2,n,\delta)}$ is a BCH code, which leads to $\mathcal{C}_{(2,n,\delta)}=\{\mathbf{0}\}$. It is obvious that $\mathcal{C}_{(2,n,\delta)}\neq\{\mathbf{0}\}$, which is a contradiction.

\noindent {\bf Case 3:} $q+1< \delta \le \frac{(q-2)(\sum_{i=0}^{m-1}q^i)}{q+1}$.  Since $q-1<\delta$, we have $q-1\in T$.
Then
$$\frac{q^{m}-q^{m-2}}{q+1}=\frac{q^{m-2}(q^2-1)}{q+1}=q^{m-2}(q-1)\in T.$$
Hence,
 $\frac{q^{m-2}-1}{q+1}=n-\frac{(q^m-q^{m-2})}{q+1} \in T^{-1}$ and $\frac{q^{m-2}-1}{q+1} \not \in T^\bot=\Bbb Z_n \setminus T^{-1}$.
From Lemma \ref{lem:122201}, we know that $\frac{\sum_{i=0}^{m-1}q^i}{q+1}\in T^\bot$
			is the coset leader modulo $n$.
It is clear that $\frac{q^{m-2}-1}{q+1} <\frac{\sum_{i=0}^{m-1}q^i}{q+1}$.		
 It then follows
			that there is no integer $J \ge 1$ such that $T^\bot=C_0 \cup C_1 \cup \cdots \cup C_{J-1}$, i.e., $\mathcal C_{(q,n,\delta)}^\bot$ is not a BCH code with respect to $\beta$.

\noindent {\bf Case 4:} $\frac{(q-2)(\sum_{i=0}^{m-1}q^i)}{q+1} < \delta \leq\delta_1$. It is clear that $\delta_1\notin T$, which implies that $n-\delta_1\notin T^{-1}$, i.e., $n-\delta_1\in Z_n\setminus T^{-1}=T^{\perp}$.
Obviously, we have
$$n-\delta_1=\frac{q^{m-1}+q^{\frac{m}{2}}}{q+1}=\frac{q^{\frac{m}{2}}(q^{\frac{m}{2}-1}+1)}{q+1}\in C_{\frac{q^{\frac{m}{2}-1}+1}{q+1}}$$
since $\frac{q^{\frac{m}{2}-1}+1}{q+1}$ is a coset leader modulo $n$ from  Lemma \ref{lemma:1222}.
From Lemma \ref{lem:122201}, we know that $2\notin\mathcal C_{(q,n,\delta)}^\bot.$ Then $T^\bot=C_0 \cup C_1 $ if $\mathcal C_{(q,n,\delta)}^\bot$ is a BCH code.
If $m\neq 4$, it is obvious that $C_{\frac{q^{\frac{m}{2}-1}+1}{q+1}} \not\subseteq C_0 \cup C_1$. Hence, $\mathcal C_{(q,n,\delta)}^\bot$ is not a BCH code.

If $m=4$, we obtain that $\delta=\delta_1$. It is clear that  $T=C_1 \cup C_2\cdots \cup C_{\delta_1-1}$ and $T^{\perp}=C_0 \cup C_1$. This means that $\mathcal{C}^{\perp}_{(2,n,\delta)}=\mathcal{C}_{(2,n,2,0)}$ is a BCH code with the designed distance $2$ with respect to $\beta$.

Combining all the cases above, the desired conclusion then follows.
\end{proof}

\section{BCH codes of length $n=\frac{q^m-1}{q-1}$ and its dual}\label{sec-minus}

Throughout this section, we always assume that $n=\frac{q^m-1}{q-1}$, where $q \ge 3$ is a prime power and $m \ge 4$ is an integer. In accordance with the notation specified in Section II, we first consider the parameters of the BCH code $\mathcal{C}_{(q,n,\delta_1)}$, and then show a sufficient and necessary condition for $\mathcal{C}_{(q,n,\delta)}$ being a dually-BCH code, where $2\leq \delta \leq n$ and $\delta_1$ is the largest coset leader modulo $n$.
It is clear that the defining set of $\mathcal C_{(q,n,\delta)}$ with respect to $\beta$
	is $T=C_1 \cup C_2 \cup \cdots \cup C_{\delta-1}$.
	As before, denote by $T^{\bot}$ the defining set of the dual code $\mathcal C_{(q,n,\delta)}^\bot$ with respect to $\beta$.
	Then $T^{\bot}=\Bbb Z_n \setminus T^{-1}$ and $0 \in T^\bot$.

\begin{lemma}\label{round up}
Let $q\geq 3$ be a prime power and $m\geq 4$  an integer. Let $q-1=mt_1+t_2$, where $t_1\geq 0$ and $m>t_2\geq 0$. Assume that $\Upsilon=\{\lceil \frac{m\gamma}{t_2}-1\rceil,\,\,\gamma=1,2,\cdots,t_2\}$ if $t_2\neq0$. Let $$\sum_{t=1}^{q-1}q^{\lceil\frac{mt}{q-1}-1\rceil}=a_{m-1}q^{m-1}+a_{m-2}q^{m-2}+\cdots+a_1q+a_0.$$  If $t_2=0$, then $a_i=\frac{q-1}{m}$ for all $i\in[0,m-1]$. If $t_2\neq0$,  then $a_i=\lceil \frac{q-1}{m}\rceil$ if $i\in \Upsilon$, $a_i=\lfloor\frac{q-1}{m}\rfloor$ if $i\in[0,m-1]\setminus \Upsilon$ and $\sum_{i=0}^{m-1}a_i=q-1$.
\end{lemma}
\begin{proof}
We only prove the lemma for the case that $q-1\geq m$. The case $q-1<m$ can be shown similarly and we omit the details.

If $t_2=0$, it is clear that $a_0=a_1=\cdots=a_{m-1}=\frac{q-1}{m}$. The desired conclusion then follows. If $t_2\neq 0$,
let $it_2=mu_i+v_i$, where $i\in[0,m-1]$, $0\leq u_i\leq t_2-1$ and $0\leq v_i<m$. To determine the value of $a_i$ for $i\in [0,m-1]$, we need to consider the possible values of $\lceil\frac{mt}{q-1}\rceil$ for $t \in [1,q-1]$. There are four cases.

\noindent{\bf Case 1:} $t \in \left[1,it_1+u_i\right]$. It is clear that
\begin{equation*}
\begin{split}
\left\lceil\frac{mt}{q-1}-1\right\rceil&\leq \left\lceil\frac{m(it_1+u_i)}{q-1}-1\right\rceil.\\
\end{split}
\end{equation*}
Since $mt_1=q-1-t_2$ and $mu_i=it_2-v_i$, we obtain
\begin{equation*}
\begin{split}
\left\lceil i-1+\frac{mu_i-it_2}{q-1}\right\rceil=\left\lceil i-1-\frac{v_i}{q-1}\right\rceil= i-1.
\end{split}
\end{equation*}
In this case, we have $\left\lceil\frac{mt}{q-1}-1\right\rceil\leq i-1$.

\noindent{\bf Case 2:} $t \in \left[it_1+u_i+1, (i+1)t_1+u_i\right]$.
It is clear that $t$ can be expressed as $t=it_1+u_i+ g$, where $1\leq  g < t_1$.
Since $mt_1=q-1-t_2$ and $mu_i=it_2-v_i$, we have
\begin{equation*}
\begin{split}
\left\lceil\frac{mt}{q-1}-1\right\rceil=\left\lceil i-1+\frac{mu_i+mg-it_2}{q-1}\right\rceil=\left\lceil i-1+\frac{mg-v_i}{q-1}\right\rceil.
\end{split}
\end{equation*}
It is clear that $0<mg-v_i<q-1$, then $\lceil\frac{mt}{q-1}-1\rceil=i$.

\noindent{\bf Case 3:} $t=it_1+u_i+t_1+1$. From $mt_1=q-1-t_2$ and $mu_i=it_2-v_i$, we have
$$\frac{mt}{q-1}=i+\frac{t_1m+m-v_i}{q-1},$$
Then $\lceil \frac{mt}{q-1}-1\rceil=i$ if $m-v_i\leq t_2$, and $\lceil \frac{mt}{q-1}-1\rceil=i+1$ if $m-v_i> t_2$.

\noindent{\bf Case 4:} $t \in \left[it_1+u+t_1+2,q-1\right]$. Similar as above,  it is easy to get that
\begin{equation*}
\begin{split}
\left\lceil\frac{mt}{q-1}-1\right\rceil\geq i+1.
\end{split}
\end{equation*}
From above four cases, we know that
$$\left\{\begin{array}{ll}
           \left\lceil\frac{mt}{q-1}-1\right\rceil\leq i-1, &\text{if}\,\,t\in [1,it_1+u_i],\\
            \left\lceil\frac{mt}{q-1}-1\right\rceil= i, &\text{if}\,\,t \in \left[it_1+u_i+1, (i+1)t_1+u_i\right], or\,\, t=it_1+u_i+t_1+1\,\, and \,\,m-v_i\leq t_2,\\
            \left\lceil\frac{mt}{q-1}-1\right\rceil= i+1, &\text{if}\,\,t=it_1+u_i+t_1+1\,\, and \,\,m-v_i> t_2,\\
            \left\lceil\frac{mt}{q-1}-1\right\rceil\geq i+1, &\text{if}\,\,t \in  \left[it_1+u_i+t_1+2,q-1\right].
         \end{array}
\right.
$$
When $i$ runs over $[0,m-1]$, note that $t_1=\lfloor\frac{q-1}{m}\rfloor$ and $t_1+1=\lceil\frac{q-1}{m}\rceil$, it is easy to get that
$$a_{i}=\left\{\begin{array}{cc}
           \lceil \frac{q-1}{m}\rceil, &if\,\,0<m-v_i\leq t_2,\\
           \lfloor \frac{q-1}{m}\rfloor, & if\,\,t_2 < m-v_i
         \end{array}
\right.
$$
since the number of $t$ in the range $\left[it_1+u_i+1, (i+1)t_1+u_i\right]$ is $\lfloor \frac{q-1}{m}\rfloor$.
Then $a_{i}= \lceil \frac{q-1}{m}\rceil$ if and only if $0<m-v_i\leq t_2$. Since $v_i=it_2-mu_i$, we have $a_{i}= \lceil \frac{q-1}{m}\rceil$ if and only if
$$\frac{m(u_i+1)}{t_2}-1\leq i<\frac{m(u_i+1)}{t_2},$$
which implies that
$i=\lceil\frac{m(u_i+1)}{t_2}-1\rceil$. This means that $a_{i}= \lceil \frac{q-1}{m}\rceil$ if and only if $i\in\Upsilon$. The desired conclusion then follows.
\end{proof}

\begin{lemma}\label{morethanzeroproof}
Let $t_2\neq 0$ and $N_{\gamma}^{\xi}=\lceil\frac{m\gamma}{t_2}-1\rceil-\lceil\frac{m(\gamma-\xi)}{t_2}-1\rceil$, where $1\leq\xi\leq t_2$ and $1\leq\gamma\leq t_2$. Then the following statements hold.
\begin{enumerate}
\item[1.] If $ t_2\,|\,m$, then $N_i^{\xi}=N_j^{\xi}$, where $1\leq i,j \leq t_2$.
\item[2.] If $ t_2\,\nmid\,m$, then $N_{\gamma}^{\xi}-N_{t_2}^{\xi}=0$ or $1$. Moreover, there exists $1\leq \xi_0\leq t_2$ such that
   $N_{\gamma}^{\xi_0}= N_{t_2}^{\xi_0}.$
\end{enumerate}
\end{lemma}

\begin{proof}
If $ t_2\,|\,m$, the desired conclusion follows from the definition of $N_{\gamma}^{\xi}$, directly. Next, we give the proof for the case $ t_2\,\nmid\,m$.
By the definition of $N_{\gamma}^{\xi}$, it is easy to get that
$$ \frac{m\xi_0}{t_2}-1= \frac{m\gamma}{t_2}-1-\left(\frac{m(\gamma-\xi)}{t_2}\right)<N_{\gamma}^{\xi}<\frac{m\gamma}{t_2}-\left(\frac{m(\gamma-\xi)}{t_2}-1\right)=\frac{m\xi_0}{t_2}+1.$$
Similarly, we have
$$\frac{m\xi}{t_2}-1<N_{t_2}^{\xi}<\frac{m\xi}{t_2}.$$
Then $N_{\gamma}^{\xi}-N_{t_2}^{\xi}=0$ or $1$. When $\xi=t_2$ we have $N_{\gamma}^{\xi}=N_{t_2}^{\xi_0}=m$, then there must exist $\xi_0\in[1,t_2]$ such that
   $N_{\gamma}^{\xi_0}= N_{t_2}^{\xi_0}.$ 	The proof is then completed.
\end{proof}

\begin{lemma}\label{detal}
Let $q\geq 3$ be a prime power and $m\geq 4$ be an integer. Then
$\theta=q^{m-1}-1-\frac{(\sum_{t=1}^{q-2}q^{\lceil\frac{mt}{q-1}-1\rceil}-q+2)}{q-1}$
is a coset leader modulo $n$.
\end{lemma}
\begin{proof}
Note that
	\begin{equation}\label{eq:deta418}
		\theta q^i \pmod{ n} \geq\theta \ \ \Leftrightarrow \ \ \theta (q-1)q^i \pmod{q^m-1} \geq\theta(q-1)
	\end{equation}
for any $1\leq i\leq m-1$. Below we will prove that
$$\theta(q-1)q^i \pmod{q^m-1} \geq\theta(q-1)$$
holds for any $1\leq i\leq m-1$.
It is clear that
 $$\theta=q^{m-1}-1-\frac{(\sum_{t=1}^{q-2}q^{\lceil\frac{mt}{q-1}-1\rceil}-q+2)}{q-1}=
\frac{q^m-\sum_{t=1}^{q-1}q^{\lceil\frac{mt}{q-1}-1\rceil}-1}{q-1}.$$
Then by Lemma \ref{round up} we have
\begin{equation}\label{eq:deta418-01}
\theta(q-1)=q^{m}-a_{m-1}q^{m-1}-a_{m-2}q^{m-2}-\cdots-a_{1}q-a_{0}-1
\end{equation}
and
\begin{equation}\label{eq:deta418-02}
\begin{split}
\theta(q-1)q^{i} \pmod {q^m-1}= & q^m-a_{m-i-1}q^{m-1}-\cdots-a_1q^{i+1}-a_0q^i-\\
&a_{m-1}q^{i-1}-\cdots-a_{m+1-i}q-a_{m-i}-1.
\end{split}
\end{equation}

For the sake of narrative, we denote $\theta'=(q-1)\theta$. If $m|(q-1)$,  from Lemma \ref{round up} we know that $a_{i}=\frac{q-1}{m}$ for all $i\in[0,m-1]$, then $\theta'q^{i}$ (mod $q^m-1$) $=\theta'$. From (\ref{eq:deta418}), (\ref{eq:deta418-01}) and  (\ref{eq:deta418-02}), we see that $\theta$ is a coset leader modulo $n$.
If $m\nmid (q-1)$, we have the following three cases.

\noindent{\bf Case 1:} $i\in[1,m-1]\setminus \Upsilon$. From Lemma \ref{round up}, we know that $a_{i}=\lfloor\frac{q-1}{m}\rfloor$. Then $\theta' q^{m-i-1} \pmod {q^m-1} >\theta'$ since $a_{m-1}=\lceil\frac{q-1}{m}\rceil=\lfloor\frac{q-1}{m}\rfloor+1$.

\noindent{\bf Case 2:} $i\in\Upsilon$ and $t_2|m$. Put $1\leq h \leq m$. From Lemma \ref{round up}, we know that
$$a_{h-1}=\left\{\begin{array}{cc}
           \lceil \frac{q-1}{m}\rceil, &if\,\,\frac{m}{t_2}\,|\,h,\\
           \lfloor \frac{q-1}{m}\rfloor, & if\,\,\frac{m}{t_2}\nmid h.
         \end{array}
\right.
$$
Let $h_1= (\overline{h-i})_m$,
then $a_{h}=a_{h_1}$. Hence, the sequences of
$$(a_{m-1},a_{m-2},\cdots,a_1,a_0)_q$$
and
$$(a_{m-i-1},a_{m-i-2},\cdots,a_{m+1-i},a_{m-i})_q$$
are the same. From (\ref{eq:deta418-01}) and  (\ref{eq:deta418-02}), we have
 $\theta' q^{i} \pmod {q^m-1} =\theta'$.

 \noindent{\bf Case 3:} $i\in\Upsilon$ and $t_2\nmid m$. From Lemma \ref{round up}, let $0\leq h \leq m-1$, we know that $a_h=\lceil \frac{q-1}{m}\rceil$ if and only if $h=\lceil \frac{m(t_2-l)}{t_2}-1\rceil$ and  $a_h=\lfloor \frac{q-1}{m}\rfloor$ for the other values of $h$,  where $l \in [0,t_2-1]$.
  It is easy to check that
$$m-1- \left\lceil\frac{m(t_2-l)}{t_2}-1\right\rceil=m-1-\lceil m-1-\frac{ml}{t_2}\rceil=\left\lfloor \frac{ml}{t_2}\right\rfloor.$$
Let $a=\lceil \frac{q-1}{m}\rceil$ and $b=\lfloor \frac{q-1}{m}\rfloor$, then the sequence of
$$(a_{m-1},a_{m-2},\cdots,a_1,a_0)_q$$
can be expressed as
$$(a\underbrace{\underbrace{\underbrace{\underbrace{b,\ldots,b}_{\lfloor \frac{m}{t_2}\rfloor-1},a,b,\cdots,b}_{\lfloor \frac{2m}{t_2}\rfloor-1},a,\cdots}_{\cdots},a,\cdots b}_{\lfloor \frac{m(t_2-1)}{t_2}\rfloor-1},a, b,\cdots,b)_q.$$
Since $i\in\Upsilon$, we can assume that $m-1-i=\lceil \frac{m\gamma_1}{t_2}-1\rceil$, where $\gamma_1 \in [1,t_2-1]$.
Moreover, we have
$$\left\lceil \frac{m\gamma_1}{t_2}-1\right\rceil-\left\lceil \frac{m(\gamma_1+l)}{t_2}-1\right\rceil\leq\left\lfloor \frac{ml}{t_2}\right\rfloor.$$
Then the sequence of
$$(a_{m-i-1},a_{m-i-2},\cdots,a_1,a_0,a_{m-1},a_{m-2},\cdots,a_{m-i})_q$$
can be expressed as
$$(a\underbrace{\underbrace{\underbrace{\underbrace{b,\ldots,b}_{\leq \lfloor \frac{m}{t_2}\rfloor-1},a,b,\cdots,b}_{\leq \lfloor \frac{2m}{t_2}\rfloor-1},a,\cdots}_{\cdots},a,\cdots b}_{\leq \lfloor \frac{m(t_2-1)}{t_2}\rfloor-1},a, b,\cdots,b)_q.$$
Hence, from (\ref{eq:deta418-01}) and  (\ref{eq:deta418-02}), we obtain
 $\theta' q^{i} \pmod {q^m-1}\geq\theta'$.

From Cases 1,  2 and  3, we have  $\theta' q^{i} \pmod {q^m-1}\geq\theta'$ for any $i\in[0,m-1]$. Then $\theta$ is a coset leader modulo $n$ from (\ref{eq:deta418}). This completes the proof.
\end{proof}

Let
\begin{equation}\label{eq:M}
M=q^m-\sum_{t=1}^{q-1}q^{\lceil\frac{mt}{q-1}-1\rceil}-1+\mu(q-1)
\end{equation}
 and  $0< \mu<\frac{\sum_{t=1}^{q-1}q^{\lceil\frac{mt}{q-1}-1\rceil}}{q-1}$.
By the definition of $\mu$, we know that $\mu$  can be expressed as
$$\mu=b_{m-2}q^{m-2}+b_{m-3}q^{m-3}+\cdots+b_1q+b_0,$$
where $b_0,\cdots,b_{m-2} \in [0,q-1]$. Then
$$(q-1)\mu=b_{m-2}q^{m-1}+(b_{m-3}-b_{m-2})q^{m-2}+\cdots+(b_1-b_2)q^2+(b_0-b_1)q-b_0.$$
By Lemma~\ref{round up}, we have
\begin{equation}\label{eq:MM}
M=q^m-(a_{m-1}-b_{m-2})q^{m-1}-(a_{m-2}+b_{m-2}-b_{m-3})q^{m-2}-\cdots-( a_1 +b_1-b_0)q-a_0-b_0-1.
\end{equation}
We next prove that $M$ is a not a coset leader modulo $q^m-1$ from the following three lemmas.

\begin{lemma}\label{eq:0713}
Let the notation be given as above. If $b_{m-2}> 0$, $0\leq a_0+b_0<q$ and $0\leq a_i+b_i-b_{i-1}<q$ for $i \in[1,m-2]$, then $M$ is not a coset leader modulo $q^m-1$.
\end{lemma}
\begin{proof}
If $b_{m-2}=a_{m-1}$, from Lemma~\ref{round up}, we know that $$b_{m-2}+a_{m-2}-b_{m-3}=a_{m-1}+a_{m-2}-b_{m-3}\leq 2 \lceil \frac{q-1}{m}\rceil-b_{m-3}\leq \frac{q+1}{2}\leq q-1.$$
It is easy to see that the equality can not hold simultaneously. Hence,
$M>(q-1)q^{m-1}$ from (\ref{eq:MM}). By Lemma \ref{lem:qm1q1}, we know  that $M$ is not a coset leader modulo $q^m-1$.

We now prove that $M$ is not a coset leader modulo $q^m-1$ when $b_{m-2}<a_{m-1}$.
 From Lemma~\ref{round up} we know that $a_{m-1}=\lceil\frac{q-1}{m}\rceil$. Let
\begin{equation}\label{a}
 a=a_{m-1}-b_{m-2},
\end{equation}
then $0< a<\lceil\frac{q-1}{m}\rceil$.
If there exists a positive integer $i_0\in[1,m-2]$ satisfying $a_{i_0}+b_{i_0}-b_{i_0-1}>a$, then
$ Mq^{i_0} \pmod{q^m-1}<M.$

If for all $i\in[0,m-2]$ we have $a_i+b_i-b_{i-1}\leq a$, i.e.,
\begin{equation}\label{case1array}
  \left\{\begin{array}{l}
           a_{m-2}+b_{m-2}-b_{m-3}\leq a,\\
            a_{m-3}+b_{m-3}-b_{m-4}\leq a,\\
            \vdots\\
            a_1+b_1-b_0\leq a,
         \end{array}
\right.
\end{equation}
then by (\ref{a}) and (\ref{case1array}) we have
\begin{equation*}
  a_{m-1}+a_{m-2}+\cdots+a_2+a_1-b_0\leq (m-1)a,
\end{equation*}
i.e.,
\begin{equation}\label{eq:o6o7}
 a_{m-1}+a_{m-2}+\cdots+a_2+a_1+a_0-a_0-b_0\leq (m-1)a.
\end{equation}
Since $\sum_{i=0}^{m-1}a_i=q-1$ and $a<\lceil\frac{q-1}{m}\rceil$, from (\ref{eq:o6o7}) we have
\begin{equation*}
 b_0+a_0\geq q-1-(m-1)a>a.
\end{equation*}
From (\ref{eq:MM}) we know that $ Mq^{m-1} \pmod{q^m-1}<M.$ The desired conclusion then follows.
\end{proof}

\begin{lemma}\label{eq:071301}
Let the notation be given as above. If $b_{m-2}= 0$, $0\leq a_0+b_0<q$ and $0\leq a_i+b_i-b_{i-1}<q$ for $i \in[1,m-2]$, then $M$ is not a coset leader modulo $q^m-1$.
\end{lemma}
\begin{proof}
If $M$ is a coset leader modulo $q^m-1$, then $M\leq Mq^i \pmod{q^m-1}$ for $i \in [1,m-1]$. From Lemma \ref{round up} and the definition of $M$, we have
\begin{equation}\label{case2array11}
  \left\{\begin{array}{l}
a_{m-2}-b_{m-3}\leq \lceil\frac{q-1}{m}\rceil,\\
            b_{m-3}+a_{m-3}-b_{m-4}\leq \lceil\frac{q-1}{m}\rceil,\\
            \vdots\\
            b_1+a_1-b_0\leq \lceil\frac{q-1}{m}\rceil,\\
            a_0+b_0\leq \lceil\frac{q-1}{m}\rceil.
         \end{array}
\right.
\end{equation}
If $m|(q-1)$, we have $\lceil\frac{q-1}{m}\rceil=\frac{q-1}{m}$.
Then from (\ref{case2array11}) and $a_i=\frac{q-1}{m}$ for all $i\in[0,m-1]$, we have
$$b_0=b_1=\cdots=b_{m-3}=0.$$
Combining with $b_{m-2}=0$, we obtain that $\mu=0$, which is contradictory to $\mu>0$. Then $M$ is not a coset leader modulo $q^m-1$.

In the following, we prove that $M$ is not a coset leader modulo $q^m-1$  if $m\nmid (q-1)$.
In this case, we have $\lceil\frac{q-1}{m}\rceil-\lfloor\frac{q-1}{m}\rfloor=1.$
By Lemma \ref{round up}, we know that $a_0=\lfloor\frac{q-1}{m}\rfloor$. Then $b_0=0$ or $b_0=1$.

\noindent{\bf Case 1:} $b_0=1$. Recall that
\begin{equation}\label{eq:0625}
   M=q^m-(a_{m-1}-b_{m-2})q^{m-1}-(b_{m-2}+a_{m-2}-b_{m-3})q^{m-2}-\cdots-(b_1+a_1-b_0)q-a_0-b_0-1.
\end{equation}
Then
\begin{equation}\label{eq:062501}
\begin{split}
  q^{m-1}M \pmod {q^m-1}= & q^m-(b_0+a_0)q^{m-1}-(a_{m-1}-b_{m-2})q^{m-2}-\cdots-(b_2+a_2-b_1)q\\
  &-(b_1+a_1-b_0)-1.
\end{split}
\end{equation}
 If $M\leq q^{m-1}M$, by comparing (\ref{eq:0625}) and (\ref{eq:062501}), we know that
\begin{equation*}
a_{m-1}=a_{m-2}=\cdots=a_1=\left\lceil\frac{q-1}{m}\right\rceil\,\,\text{and}\,\,b_0=b_1=\cdots=b_{m-3}=0
\end{equation*}
since $b_{m-2}=0$ and $\lceil\frac{q-1}{m}\rceil\leq a_i\leq\lfloor\frac{q-1}{m}\rfloor$. Then we obtain that $\mu=0$, which is contradictory to $\mu>0$. Hence, we have $M> q^{m-1}M \pmod {q^m-1}$.

\noindent{\bf Case 2:} $b_0=0$.
From (\ref{case2array11}), we know that $b_1=0$ or $b_1=1$. If $b_1=0$, then from (\ref{case2array11}) we know that $b_2=0$ or $b_2=1$. Continue this work, we can obtain that there exists $i_1$ such that $b_{i_1}=1$ and $b_{i_1-l}=0$, where $i_1\in[1,m-3]$ and $l\in [1,i_1-1]$.

Recall that $\Upsilon$ is defined in Lemma \ref{round up}. If $i_1\in\Upsilon $, then we have $b_{i_1-1}=0$, $b_{i_1}=1$ and $a_{i_1}=\lceil\frac{q-1}{m}\rceil$, which is contradictory to $a_{i_1}+b_{i_1}-b_{i_1-1}\leq \lceil\frac{q-1}{m}\rceil$. Hence, we obtain that $i_1\in[0,m-1]\setminus \Upsilon $. This means that there is $\gamma_2\in[0,t_2-1]$ satisfying $\lceil\frac{\gamma_2m}{t_2}-1\rceil< i_1<\lceil\frac{(\gamma_2+1)m}{t_2}-1\rceil$. For the sake of narrative, we assume that $b_{m-1}=b_{-1}=0$ in the following of this proof.

 \noindent{\bf Subcase 1}: $\lceil\frac{\gamma_2m}{t_2}-1\rceil< i_1<\lceil\frac{(\gamma_2+1)m}{t_2}-1\rceil-1$.
Let $\xi=m- i_1-1+\lceil\frac{ \gamma_2m}{t_2}-1\rceil$. It is clear that $M$ and $q^{m-i_1-1}M$ can be expressed as
\begin{equation}\label{subcase1M}
\begin{split}
  M= & q^m-(a_{m-1}-b_{m-2})q^{m-1}-(a_{m-2}+b_{m-2}-b_{m-3})q^{m-2}-\cdots-(b_{\xi}+a_{\xi}-b_{\xi-1})q^{\xi}-\cdots \\
    &-(b_1+a_1-b_0)q-a_0-b_0-1,
\end{split}
\end{equation}
and
\begin{equation}\label{subcase1qm1}
\begin{split}
   q^{m-i_1-1}M\pmod{q^m-1}=& q^m-(a_{i_1}+b_{i_1}-b_{i_1-1})q^{m-1}-(a_{ i_1-1}+b_{ i_1-1}-b_{ i_1-2})q^{m-2}-\cdots\\
   &-(a_{\lceil\frac{(\overline{\gamma_2-1})_{t_2}m+m}{t_2}-1\rceil}+b_{\lceil\frac{(\overline{\gamma_2-1})_{t_2}m+m}{t_2}-1\rceil}-b_{\lceil\frac{ (\overline{\gamma_2-1})_{t_2}m+m}{t_2}-2\rceil})q^{\xi}\\
&-\cdots-(b_{i_1}+a_{i_1}-b_{i_1-1})q-(b_{i_1-1}+a_{i_1-1}-b_{i_1-2})-1.
\end{split}
\end{equation}
Since $b_{m-2}=0$, $b_{i_1}=1$ and $b_{i_1-l}=0$ for $l\in [1,i_1-1]$, we obtain that
$$\left\lceil\frac{q-1}{m}\right\rceil=a_{i_1}+b_{i_1}-b_{i_1-1}=a_{m-1}-b_{m-2}.$$
It is clear that
$$m-1-\left(m-i_1-1+\left\lceil\frac{ \gamma_2m}{t_2}-1\right\rceil\right)=i_1-\left\lceil\frac{ \gamma_2m}{t_2}-1\right\rceil\leq N_{ \gamma_2+1}^1 -2\leq N_{t_2}^1-1.$$
Then
\begin{equation*}
a_{m-2}=a_{m-3}=\cdots=a_{\xi+1}=\left\lfloor \frac{q-1}{m}\right\rfloor\,\, \text{and}\,\,
a_{i_1}=a_{i_1-1}=\cdots= a_{\lceil\frac{ {\gamma_2}m}{t_2}\rceil}=\left\lfloor \frac{q-1}{m}\right\rfloor
\end{equation*}
since $0\leq N_{\gamma_2+1}^1-N_{t_2}^1\leq 1$ and $\lceil\frac{\gamma_2m}{t_2}-1\rceil<i_1<\lceil\frac{(\gamma_2+1)m}{t_2}-1\rceil-1$.
 If there exists  $l_1\in [2,m-2-\xi]$ such that
\begin{equation}\label{eq:minl}
  a_{m-l_1-1}+b_{m-l_1-1}-b_{m-l_1-2}>a_{i_1-l_1}+b_{i_1-l_1}-b_{i_1-l_1-1},
\end{equation}
let $l_2$ be the least integer in the range $[2,m-2-\xi]$ such that (\ref{eq:minl}) holds, then we have $b_{m-l_2-1}>0$ and
\begin{equation}\label{eq:minl1}
  a_{m-1-l_3}+b_{m-1-l_3}-b_{m-2-l_3}\leq a_{i_1-l_3}+b_{i_1-l_3}-b_{i_1-l_3-1}
\end{equation}
for all $0\leq l_3< l_2$. If all of the equals in (\ref{eq:minl1}) hold, since $b_{m-2}=0$ and $b_{i_1-l}=0$ for all $l\in[1,i_1-1]$, we know that $b_{m-3}=b_{m-4}=\cdots=b_{m-l_2-1}=0$, which is contradictive with $b_{m-l_2-1}>0$. Then at least one of equals  in (\ref{eq:minl1}) does not hold, we obtain that $q^{m-i_1-1}M\pmod{q^m-1}< M$ from (\ref{subcase1M}) and (\ref{subcase1qm1}).

If there does not exist $l_1\in [2,m-2-\xi]$ such that (\ref{eq:minl}) holds, then we have
\begin{equation}\label{case2array1}
  \left\{\begin{array}{l}
           \lfloor \frac{q-1}{m}\rfloor=a_{i_1-1}+b_{i_1-1}-b_{i_1-2}\geq b_{m-2}+a_{m-2}-b_{m-3},\\
            \lfloor \frac{q-1}{m}\rfloor=a_{i_1-2}+b_{i_1-2}-b_{i_1-3}\geq b_{m-3}+a_{m-3}-b_{m-4},\\
            \vdots\\
           \lfloor \frac{q-1}{m}\rfloor=a_{\lceil\frac{ \gamma_2m}{t_2}\rceil}+b_{\lceil\frac{\gamma_2m}{t_2}\rceil}
           -b_{\lceil\frac{ \gamma_2m}{t_2}-1\rceil}
           \geq b_{\xi+1}+a_{\xi+1}-b_{\xi}.
         \end{array}
\right.
\end{equation}
If there is one of equals  in (\ref{case2array1}) that does not hold, then $q^{m-i_1-1}M\pmod{q^m-1}< M$ from (\ref{subcase1M}) and (\ref{subcase1qm1}).

If all the {equals} in (\ref{case2array1}) holds, then $b_{m-3}=b_{m-4}=\cdots=b_\xi=0$. Hence,
$$a_\xi+b_\xi-b_{\xi-1}= \left\lfloor \frac{q-1}{m}\right\rfloor-b_{\xi-1}\leq \left\lfloor \frac{q-1}{m} \right\rfloor.$$
It is clear that $a_{\lceil\frac{(\overline{\gamma_2-1})m+m}{t_2}-1\rceil}+b_{\lceil\frac{ (\overline{\gamma_2-1})m+m}{t_2}-1\rceil}-
b_{\lceil\frac{ (\overline{\gamma_2-1})m+m}{t_2}-2\rceil}=\lceil\frac{q-1}{m}\rceil$ since $a_{\lceil\frac{(\overline{\gamma_2-1})m+m}{t_2}-1\rceil}=\lceil\frac{q-1}{m}\rceil$ and $b_{\lceil\frac{ (\overline{\gamma_2-1})m+m}{t_2}-1\rceil}=
b_{\lceil\frac{ (\overline{\gamma_2-1})m+m}{t_2}-2\rceil}=0$. Then $q^{m-i_1-1}M\pmod{q^m-1}< M$ from (\ref{subcase1M}) and (\ref{subcase1qm1}).

\noindent{\bf Subcase 2:}  $i_1=\lceil\frac{(\gamma_2+1)m}{t_2}-1\rceil-1$. From Lemma \ref{morethanzeroproof}, there exists $\xi\in[1,t_2]$ such that
   $N_{\gamma_2}^{\xi}= N_{t_2}^{\xi}.$
Assume that $\xi_1\in [1,t_2]$ is the minimum value such that $ N_{\gamma_2}^{\xi_1}=N_{t_2}^{\xi_1}.$

If there exists  $l\in [2,i_1-\lceil\frac{(\gamma_2-\xi_1)m}{t_2}-1\rceil-1]$ such that $a_{m-1-l}+b_{m-1-l}-b_{m-2-l}>a_{(\overline{i_1-l})_{m}}+b_{(\overline{i_1-l})_{m}}-b_{(\overline{i_1-l-1})_m}$,
similar to the discuss of (\ref{eq:minl}) and (\ref{eq:minl1}), then there is $1\leq l_4<l$ such that $a_{m-1-l_4}+b_{m-1-l_4}-b_{m-2-l_4}<a_{(\overline{i_1-l_4})_{m}}+b_{(\overline{i_1-l_4})_m}-b_{(\overline{i_1-l_4-1})_m}$. Hence, we obtain $q^{m-i_1-1}M\pmod{q^m-1}< M$.

If there does not exist $l\in [2,i_1-\lceil\frac{(\gamma_2-\xi_1)m}{t_2}-1\rceil-1]$ such that $a_{m-1-l}+b_{m-1-l}-b_{m-2-l}>a_{(\overline{i_1-l})_m}+b_{(\overline{i_1-l})_m}-b_{(\overline{i_1-l-1})_m}$, then we have
\begin{equation}\label{eq:upeai}
  a_{(\overline{i_1-l})_m}+b_{(\overline{i_1-l})_m}-b_{(\overline{i_1-l-1})_m}\geq a_{m-1-l}+b_{m-1-l}-b_{m-2-l}
\end{equation}
for all $l\in [2,i_1-\lceil\frac{(\gamma_2-\xi_1)m}{t_2}-1\rceil-1]$.
 If one of equals  in (\ref{eq:upeai}) does not hold, then $q^{m-i_1-1}M\pmod{q^m-1}< M$.

If all the equals in (\ref{eq:upeai}) holds, then $b_{m-3}=b_{m-4}=\cdots=b_{\eta}=0$, where $\eta=m-1-i_1+\lceil \frac{(\gamma_2-\xi_1)m}{t_2}\rceil$. Hence,
\begin{equation*}
  a_{\eta+1}+b_{\eta+1}-b_{\eta}=\left\lfloor \frac{q-1}{m}\right\rfloor+0-b_{\eta}{\red =} \left\lfloor \frac{q-1}{m}\right\rfloor.
\end{equation*}
Since $a_{\lceil\frac{(\overline{\gamma_2-\xi_1-1})_{t_2}m+m}{t_2}-1\rceil}=\lceil\frac{q-1}{m}\rceil$ and $b_{\lceil\frac{(\overline{\gamma_2-\xi_1-1})_{t_2}m+m}{t_2}-1\rceil}=b_{\lceil\frac{(\overline{\gamma_2-\xi_1-1})_{t_2}m+m}{t_2}-2\rceil}=0$, we have
\begin{equation*}
 a_{\lceil\frac{(\overline{\gamma_2-\xi_1-1})_{t_2}m+m}{t_2}-1\rceil}+b_{\lceil\frac{(\overline{\gamma_2-\xi_1-1})_{t_2}m+m}{t_2}-1\rceil}
 -b_{\lceil\frac{(\overline{\gamma_2-\xi_1-1})_{t_2}m+m}{t_2}-2\rceil}=\left\lceil\frac{q-1}{m}\right\rceil,
\end{equation*}
 then
$$ a_{\eta+1}+b_{\eta+1}-b_{\eta}<a_{\lceil\frac{(\overline{\gamma_2-\xi_1-1})_{t_2}m+m}{t_2}-1\rceil}+b_{\lceil\frac{(\overline{\gamma_2-\xi_1-1})_{t_2}m+m}{t_2}-1\rceil}
 -b_{\lceil\frac{(\overline{\gamma_2-\xi_1-1})_{t_2}m+m}{t_2}-2\rceil}.$$
Hence, we obtain $q^{m-i_1-1}M\pmod{q^m-1}< M$.

Hence, from Cases 1 and  2,  we know that there always exists $i \in [1,m-1]$ such that $M\geq Mq^i \pmod{q^m-1}$. Hence, $M$ is not a coset leader modulo $q^m-1$. The desired conclusion then follows.
\end{proof}

\begin{lemma}\label{eq:071302}
Let the notation be given as above. If there exists $i\in[1,m-2]$ such that $a_i+b_i-b_{i-1}<0$ or $a_i+b_i-b_{i-1}\geq q$, then $M$ is not a coset leader modulo $q^m-1$.
\end{lemma}
\begin{proof}
If $a_i+b_i-b_{i-1}\geq 0$ for all $i\in[1,m-2]$ and there exists a positive integer $i_2\in[1,m-2]$ such that $a_{i_2}+b_{i_2}-b_{i_2-1}\geq q$, i.e.,
\begin{equation}\label{eq:geqq}
  \left\{\begin{array}{l}
           a_{m-2}+b_{m-2}-b_{m-3}\geq 0,\\
           \vdots\\
           a_{i_2+1}+b_{i_2+1}-b_{i_2}\geq 0,\\
           a_{i_2}+b_{i_2}-b_{i_2-1}\geq q,\\
         \end{array}
\right.
\end{equation}
then $b_{i_2-1}\leq \sum_{j=i_2}^{m-2}a_j+b_{m-2}-q\leq \sum_{j=i_2}^{m-2}a_j+a_{m-1}-q<0$, which is impossible. Hence, there exists $i_3\in[1,m-2]$ such that $a_{i_3}+b_{i_3}-b_{i_3-1}< 0$ if there exists $i_2\in[1,m-2]$ such that $a_{i_2}+b_{i_2}-b_{i_2-1}\geq q$.

Let $\Psi$ be a subset of $[1,m-2]$ such that $b_{i}+a_i-b_{i-1}< 0$ if $i\in \Psi$ and $a_i+b_{i}-b_{i-1}\geq 0$ if  $i\in[1,m-2]\setminus \Psi$. Let $i_4=max\{i:i\in \Psi\}$. If there exists $i$ such that $b_{i}+a_i-b_{i-1}\geq q$, we assume that $i_5=max\{i:a_i+b_i-b_{i-1}\geq q, i\in[1,m-2]\setminus \Psi\}$. If $i_5>i_4$, with a similar analysis as (\ref{eq:geqq}), it is easy to get that $b_{i_4-1}<0$, which is contradictory to $b_{i_4-1}\geq 0$.

If $b_{i_4}+a_{i_4}-b_{i_4-1}=-1$ and $i_5=i_4-1$, we know that
\begin{equation*}
  \left\{\begin{array}{l}
           a_{m-2}+b_{m-2}-b_{m-3}\geq 0,\\
           \vdots\\
           a_{{i_4+1}+b_{ i_4}+1}-b_{ i_4}\geq 0,\\
           a_{ i_4}+b_{i_4}-b_{ i_4-1}=-1,\\
           a_{ i_4-1}+b_{ i_4-1}-b_{ i_4-2}\geq q.
         \end{array}
\right.
\end{equation*}
Then we have $b_{i_4-2}\leq -q+\sum_{i=i_4-1}^{m-2}a_i+1\leq -q+\sum_{i= i_4-1}^{m-1}a_i\leq -1<0$, which is contradictory to $b_{ i_4-2}\geq 0$.

From above, in order to obtain the desired result, we only need to prove the case that $M$ is not a coset leader modulo $q^m-1$ if $b_{i_4}+a_{i_4}-b_{i_4-1}<-1$, or $i_5\neq i_4-1$, or there does not exist $i\in[1,m-2]$ such that $b_i+a_i-b_{i-1}\geq q$. In these cases, it is clear that $M$  can be expressed as
\begin{equation}\label{eq:Mm0}
M=(q-a_{m-1}+b_{m-2})q^{m-1}+(b_{m-3}-a_{m-2}-b_{m-2})q^{m-2}+\cdots+(b_0-a_1 -b_1)q-a_0-b_0-1.
\end{equation}
Then $q^{m-i_4-1}M \pmod{q^m-1}$ can be expressed as
\begin{equation}\label{eq:MI1}
\begin{split}
   q^{m-i_4-1}M\pmod{q^m-1}=& (b_{i_4-1}-a_{i_4}-b_{i_4})q^{m-1}+(b_{i_4-2}-a_{i_4-1}-b_{i_4-1})q^{m-2}+\cdots+\\
   &(-a_0-b_0)q^{m-i_4-1}+(q-a_{m-1}+b_{m-2})q^{m-i_4-2}+\\
   &\cdots+(b_{i_4}-a_{i_4+1}-b_{i_4+1})
\end{split}
\end{equation}
if $i_4\neq m-2$ and
\begin{equation*}
\begin{split}
   q^{m-i_4-1}M\pmod{q^m-1}=& (b_{i_4-1}-a_{i_4}-b_{i_4})q^{m-1}+(b_{i_4-2}-a_{i_4-1}-b_{i_4-1})q^{m-2}+\cdots+\\
   &(-a_0-b_0)q^{m-i_4-1}+(-a_{m-1}+b_{m-2})
\end{split}
\end{equation*}
if $i_4=m-2$.
We only prove the case that  $M$ is not a coset leader modulo $q^m-1$ if $i_4\neq m-2$. When $i_4=m-2$, the desired results can be shown similarly. We omit the details.

It is easy to see that $b_{i_4-1}-b_{i_4}-a_{i_4}\leq q+b_{m-2}-a_{m-1}$.  If $b_{i_4-1}-b_{i_4}-a_{ i_4}< q+b_{m-2}-a_{m-1}$, then $q^{m-i_4-1}M\pmod{q^m-1}<M$.

If $b_{i_4-1}-b_{i_4}-a_{i_4}= q+b_{m-2}-a_{m-1}$, it is clear that
\begin{equation}\label{eq:da0716}
b_{i_4-1}=q-1,\,\,b_{i_4}=0,\,\,a_{i_4}=\lfloor\frac{q-1}{m}\rfloor,\,\,b_{m-2}=0\,\, \text{and}\,\, m\nmid (q-1)
\end{equation}
since $a_{m-1}=\lceil \frac{q-1}{m}\rceil.$
Then there exists a positive integer $\gamma_3$ such that $\lceil \frac{(\gamma_3-1)m}{t_2}-1\rceil< i_4 < \lceil \frac{\gamma_3m}{t_2}-1\rceil$.
Let $1\leq s\leq i_4$ and
\begin{equation*}
   D_{(m-1-s)}  =(b_{i_4-s-1}-a_{i_4-s}-b_{i_4-s})-(b_{m-s-2}-b_{m-s-1}-a_{m-s-1}).
\end{equation*}
Clearly, from (\ref{eq:da0716}) we know that
\begin{equation*}
\begin{split}
   D_{(m-2)} & =(b_{i_4-2}-a_{i_4-1}-b_{i_4-1})-(b_{m-3}-b_{m-2}-a_{m-2}) \\
    & =b_{i_4-2}-(q-1)-b_{m-3}-a_{i_4-1}+a_{m-2}.
\end{split}
\end{equation*}
From Lemma~\ref{round up}, it is obvious that
\begin{equation}\label{0720}
-a_{i_4-1}+a_{m-2}\in \{-1,0,1\}.
\end{equation}
If $a_{m-2}-a_{i_4-1}=1$, then we have $a_{i_4-1}=\lfloor\frac{q-1}{m}\rfloor$ and $a_{m-2}=\lceil\frac{q-1}{m}\rceil$. Since $a_{i_4}=a_{i_4-1}=\lfloor\frac{q-1}{m}\rfloor$ and $a_{m-1}=a_{m-2}=\lceil\frac{q-1}{m}\rceil$, we obtain that $N_{\gamma_3}^1\geq 3$ and $N_{t_2}^1=1$, which is contradictory to Lemma \ref{morethanzeroproof}. Then from (\ref{0720}) we know that $a_{m-2}-a_{i_{4}-1} \in \{-1,0\}$. Hence, we obtain that $ D_{(m-2)}\leq 0$. If $D_{(m-2)}< 0$, then $q^{m-i_4-1}M\pmod{q^m-1}<M$.

If $D_{(m-2)}= 0$, combined with (\ref{eq:da0716}), we obtain that
\begin{equation}\label{eq:m2722}
b_{i_4-2}=q-1,\,\, b_{m-3}=0\,\, \text{and}\,\, a_{m-2}-a_{i_4-1}=0.
\end{equation}
With a similar analysis on $D_{(m-2)}$, we obtain $ D_{(m-3)}\leq 0$. If $ D_{(m-3)}< 0$, then  $q^{m-i_4-1}M\pmod{q^m-1}<M$.

 If $ D_{(m-3)}= 0$, combined with (\ref{eq:m2722}) we obtain that $b_{i_4-3}=q-1$, $b_{m-4}=0$ and $a_{m-3}-a_{i_4-2}=0$.  Continue this work, we always have $q^{m-i_4-1}M\pmod{q^m-1}<M$,
or
\begin{equation}\label{eq:07222}
  \left\{\begin{array}{l}
           b_{i_4-1}=b_{i_4-2}=\cdots=b_{0}=q-1,\\
           b_{m-2}= b_{m-3}=\cdots=b_{m-i_{4}-1}=0,\\
           \vdots\\
          a_{m-2}-a_{i_4-1}=a_{m-3}-a_{i_4-2}=a_{m-i_4}-a_1=0.
         \end{array}
\right.
\end{equation}
If (\ref{eq:07222}) holds, then
\begin{equation*}
\begin{split}
  D_{(m-i_4-1)}&=-(b_0+a_0)-(b_{m-i_{4}-2}-b_{m-i_{4}-1}-a_{m-i_{4}-1)} \\
    & =-(q-1)-\lfloor \frac{q-1}{m}\rfloor- b_{m-i_{4}-2}+a_{m-i_{4}-1}\leq -(q-2)<0.
\end{split}
\end{equation*}
Hence, $q^{m-i_4-1}M \pmod{q^m-1}<M$. Combining with all the cases, we obtain that $M$ is not a coset leader modulo $q^m-1$.
The desired conclusion then follows.
\end{proof}

\begin{proposition}\label{pro:detal}
Let $q\geq 3$ be a prime power and $m\geq 4$ be an integer. Then $\delta_1=\theta$ is the largest coset leader modulo $n$, where $\theta$ is given in Lemma \ref{detal}.
\end{proposition}
\begin{proof}
If $\delta$ is a coset leader modulo $n$, then $(q-1)\delta$ must be a coset leader modulo $q^m-1$ since
\begin{equation*}
		\theta q^i \pmod{ n} \geq\theta \ \ \Leftrightarrow \ \ \theta (q-1)q^i \pmod{q^m-1} \geq\theta(q-1)
\end{equation*}
for any $1\leq i\leq m-1$. If $\delta>\delta_1$, we know that  $(q-1)\delta$ can be written as $M$, where $M$ is given in (\ref{eq:M}). From Lemmas \ref{eq:0713}-\ref{eq:071302}, we obtain that $M$ is not a coset leader modulo $q^m-1$. Hence, $\delta$ is not a coset leader modulo $n$ if $\delta>\delta_1$, i.e., $\delta_1$ is the largest coset leader modulo $n$. The desired results then follows.
\end{proof}

\begin{lemma}\label{071303}
Let $\delta_1$ be given as in Proposition \ref{pro:detal}, then $|C_{\delta_1}|=\frac{m}{\gcd(m,q-1)}$.
\end{lemma}
\begin{proof}
It is known that ord$_n(q)=m$,
then $|C_{\delta_1}|$ is a divisor of $m$. Assume that $|C_{\delta_1}|=h$, then
\begin{equation}\label{eq:qmodd1}
n\,|\,\delta_1(q^h-1).
\end{equation}
By definition, $\delta_1$ can be written as
$$\delta_1=n-\frac{\sum_{t=1}^{q-1}q^{\lceil\frac{mt}{q-1}-1\rceil}}{q-1}.$$
Then (\ref{eq:qmodd1}) holds if and only if
$\frac{q^m-1}{q^{h}-1}\,|\,\sum_{t=1}^{q-1}q^{\lceil\frac{mt}{q-1}-1\rceil},$ i.e.,
\begin{equation*}
q^{m-h}+q^{m-2h}+\cdots+q^{m-(i-1)h}+1\,|\,\sum_{t=1}^{q-1}q^{\lceil\frac{mt}{q-1}-1\rceil},
\end{equation*}
where $i=\frac{m}{h}$.
It is easy to check that
$$\sum_{t=1}^{q-1}q^{\lceil\frac{mt}{q-1}-1\rceil}=
\left(\sum_{i_0=0}^{\gcd(q-1,m)-1}q^{\frac{mi}{\gcd(q-1,m)}}\right)\left(\sum_{i_1=1}^{\frac{(q-1)}{\gcd(m,q-1)}}q^{\lceil\frac{mt}{q-1}-1\rceil}\right).$$
In addition,
$$q^{m-h}+q^{m-2h}+\cdots+q^{m-(i-1)h}+1=\sum_{i=0}^{\gcd(q-1,m)-1}q^{\frac{mi}{\gcd(q-1,m)}}$$
if $h=\frac{m}{\gcd(q-1,m)}$. Hence, (\ref{eq:qmodd1}) holds if $h=\frac{m}{\gcd(q-1,m)}$.
This means that
\begin{equation}\label{dadf102}
|C_{\delta_1}|=h\leq \frac{m}{\gcd(q-1,m)}.
\end{equation}

We now prove that $h\geq \frac{m}{\gcd(q-1,m)}$.   Assume that there exists a positive integer $Q$  such that
\begin{equation}\label{eqfsd}
Q(q^{m-h}+q^{m-2h}+\cdots+q^{m-(i-1)h}+1)=\sum_{t=1}^{q-1}q^{\lceil\frac{mt}{q-1}-1\rceil}.
\end{equation}
Let $Q=a_iq^i+a_{i-1}q^{i-1}+\cdots+a_1q+a_0$ and $a_i+a_{i-1}+\cdots+a_0=r$. If (\ref{eqfsd}) holds, then $\frac{rm}{h}=q-1$, which implies that
$rm=h(q-1)$. Hence, we have
$$r\cdot  \frac{m}{\gcd(m,q-1)}=h \cdot \frac{q-1}{\gcd(m,q-1)}.$$
Since
$$\gcd\left(\frac{m}{\gcd(m,q-1)},\frac{q-1}{\gcd(m,q-1)}\right)=1,$$
we obtain that
 $\frac{m}{\gcd(m,q-1)}\,|\,h.$
Hence, from (\ref{dadf102}) we have $|C_{\delta_1}|= \frac{m}{\gcd(q-1,m)}$.
The desired conclusion then follows.
\end{proof}

\begin{remark}
When $q=3$, $\delta_1$ and $|C_{\delta_1}|$ have been given in \cite[Lemma 17]{Lid17}. Let $m\geq q$, $b\equiv m-1 \pmod {q-1}$. When $b=0$, $b=1$ or $b=q-2$, $\delta_1$ and $|C_{\delta_1}|$ have been given in \cite[Lemma 16]{Zhu19}. We general these results in Proposition \ref{pro:detal} and Lemma \ref{071303}.
\end{remark}

From Proposition \ref{pro:detal} and Lemma \ref{071303}, one can get the following theorem.

\begin{theorem}\label{theorem4}
 When $m\geq 3$ be an integer and $q\geq 3$ be a prime power, the BCH code $\mathcal{C}_{(q,n,\delta_1)}$ has parameters
$$\left[\frac{q^m-1}{q-1}, \frac{m}{\gcd(m,q-1)}+1, d\geq q^{m-1}-1-\frac{\sum_{t=1}^{q-2}q^{\lceil\frac{mt}{q-1}-1\rceil}-q+2}{q-1}\right].$$
\end{theorem}

\begin{example}
Let $(q,m)=(3,4)$. Then the code $\mathcal{C}_{(q,n,\delta_1)}$ in Theorem \ref{theorem2} has parameters $[40,3,\geq 25]$. This code is the best cyclic code according to \cite[P. 305]{Dingbook15} when the equality holds.
\end{example}

In the following, we present a sufficient and necessary condition for $\mathcal{C}_{(q,n,\delta)}$ being a dually-BCH code. We first give a key lemma.

	\begin{lemma} \label{proposition-(q^m-1)/(q-1)}
		Let $\delta'$ be the coset leader of $C_{n-\delta_1}$ modulo $n$. Then the following hold.
		\begin{enumerate}
			\item $\delta_1 \in T^\bot$ is a coset leader modulo $n$ if $2 \le \delta \le \delta'$.
			\item $\delta' \in T^\bot$ is a coset leader modulo $n$ if $\delta' < \delta \le \delta_1$.
		\end{enumerate}
	\end{lemma}
	\begin{proof}
Since $\delta'$ be the coset leader of $C_{n-\delta_1}$ modulo $n$, we have $C_{n-\delta'}=C_{\delta_1}.$ If $2 \le \delta \le \delta'$, i.e., $C_{\delta'} \nsubseteq T$, then $C_{n-\delta'}=C_{\delta_1} \nsubseteq T^{-1}$, $C_{\delta_1} \subseteq T^{\bot}$.
		If $\delta' < \delta \le \delta_1$, i.e., $C_{\delta_1} \nsubseteq T$, then $C_{n-\delta_1}=C_{\delta'} \nsubseteq T^{-1}$ and $C_{\delta'} \subseteq T^{\bot}$.
			\end{proof}

Denote $m=r(q-1)+s$ and $\nu=\lceil \frac {(s-1)(q-1)} s \rceil$, where $r \ge 0$ and $0 \le s \le q-2$. Note that $$\delta_1 = q^{m-1}-1-\frac{\sum\limits_{t=1}^{q-2}q^{\lceil \frac{mt}{q-1}-1 \rceil}-q+2}{q-1} =\frac{q^m-\sum\limits_{t=1}^{q-1}q^{\lceil \frac{mt}{q-1}-1 \rceil}-1}{q-1}.$$  Then
\begin{equation}\label{eq:0711}
n-\delta_1=\frac{\sum_{t=1}^{q-1}q^{\lceil \frac{mt}{q-1}-1 \rceil}}{q-1}.
\end{equation}

 Let $\delta'$ and $\delta''$ be the coset leaders of $C_{n-\delta_1}$ modulo $n$ and $C_{(q-1)(n-\delta_1)}$ modulo $(q-1)n$, respectively.
It is similar with (\ref{eq:deta418}), we have $(q-1) \mid \delta''$ and $\delta''=\delta'(q-1)$. From (\ref{eq:0711}) we have
$(q-1)(n-\delta_1)=\sum_{t=1}^{q-1}q^{\lceil \frac{mt}{q-1}-1 \rceil}$. From Lemma \ref{round up}, one can see that the $q$-adic expansion of $\delta''$ has the form $(\textbf{0}_r, 1, \cdots)_q$ if $r\geq 1$ and $\delta''$ has the form $(\lfloor \frac{q-1}{m}\rfloor, \cdots)_q$ if $r=0$.
where $\textbf{0}_r=(\underbrace{0,0,\cdots, 0}_r)_q$.
Then
 \begin{equation}\label{eqn-delta'} \delta'=\frac {\delta''} {q-1}> \frac {q^{m-r-1}-1} {q-1}.\end{equation}

With the preparations above, we now give a sufficient and necessary condition for $\mathcal{C}_{(q,n,\delta)}$ being a dually-BCH code.

\begin{theorem}\label{theorem5}
Let $n=\frac{q^m-1}{q-1}$, where $q \ge 3$ is a prime power and $m \ge 4$ is a positive integer.
Then $\mathcal C_{(q,n,\delta)}$ is a dually-BCH code if and only if
$\delta_1+1 \le \delta \le n$, where $\delta_1$ is given in Proposition \ref{pro:detal}.
\end{theorem}
\begin{proof}
When $q=3$, the result has been given in \cite[Theorem 30]{GDL21}, we only prove the result for $q>3$ in the following.

		It is clear that $0 \notin T$ and $1 \in T $, so $0 \notin T^{-1}$ and  $n-1 \in T^{-1}$. Furthermore, we have $0 \in T^{\bot}$ and $n-1 \notin T^{\bot}$,
which means that $C_0$ must be the initial cyclotomic coset of $T^\bot$. Consequently, there must be an integer $J \ge 1$
such that $T^\bot=C_0 \cup C_1 \cup \cdots \cup C_{J-1}$ if $\mathcal C_{(q,n,\delta)}$ is a dually-BCH code.
		
When $\delta_1+1 \le \delta \le n$, it is easy to see that $T^\bot = \{0\}$ and $\mathcal C_{(q,n,\delta)}^\bot$ is a BCH code with respect to $\beta$.
It remains to show that $\mathcal C_{(q,n,\delta)}^\bot$ is not a BCH code with respect to $\beta$ when $2 \le \delta \leq \delta_1$.
To this end, we show that there is no integer $J \ge 1$ such that $T^\bot=C_0 \cup C_1 \cup \cdots \cup C_{J-1}$.
Recall that $m=r(q-1)+s$, we have the following two cases.

\noindent{\bf Case 1:} $2\le \delta \le \frac{q^{r+1}-1}{q-1}$. It is easy to see that $m-r-1\geq r+1$. Then from Equation \eqref{eqn-delta'} and Proposition \ref{proposition-(q^m-1)/(q-1)} that $\delta_1\in T^\bot$
is the coset leader of $C_{\delta_1}$.
It follows from Lemma \ref{lemma-break-point-(q^m-1)/(q-1)} that $$I_{\max}:=\max\left\{I(\delta): 2\le \delta \le \frac{q^{r+1}-1}{q-1}\right\}=I(2)=\frac{q^{m-1}-1}{q-1}=(0,\underbrace{1, \dots, 1}_{m-1})_q.$$
Note that $\delta_1 = q^{m-1}-1-\frac{\sum\limits_{t=1}^{q-2}q^{\lceil \frac{mt}{q-1}-1 \rceil}-q+2}{q-1}$. It is easy to see that $(q-1)\delta_1 > q^{m-1}-1$ and
$\delta_1>\frac{q^{m-1}-1}{q-1}$. It then follows that there is no integer $J \ge 1$ such that $T^\bot=C_0 \cup C_1 \cup \cdots \cup C_{J-1}$, i.e., $\mathcal C_{(q,n,\delta)}^\bot$ is not a BCH code with respect to $\beta$.

\noindent{\bf Case 2:} If $\frac{q^{r+1}-1}{q-1} < \delta < q^{m-1}-\frac{\sum\limits_{t=1}^{q-2}q^{\lceil \frac{mt}{q-1}-1 \rceil}-q+2}{q-1}$. It then follows from Proposition \ref{proposition-(q^m-1)/(q-1)}
that $\delta' \in T^\bot$ is the coset leader of $C_{\delta'}$.
It follows from Lemma \ref{lemma-break-point-(q^m-1)/(q-1)} that $$I_{\max}:=\max\left\{I(\delta): \frac{q^{r+1}-1}{q-1} < \delta < q^{m-1}-\frac{\sum\limits_{t=1}^{q-2}q^{\lceil \frac{mt}{q-1}-1 \rceil}-q+2}{q-1}\right\}=\frac{q^{m-r-1}-1}{q-1}.$$
We deduce from Equation \eqref{eqn-delta'} that
$\delta'>\frac{q^{m-r-1}-1}{q-1}$. It then follows from Lemma \ref{lemma-break-point-(q^m-1)/(q-1)}
that there is no integer $J \ge 1$ such that $T^\bot=C_0 \cup C_1 \cup \cdots \cup C_{J-1}$, i.e., $\mathcal C_{(q,n,\delta)}^\bot$ is not a BCH code with respect to $\beta$.		
	\end{proof}

\begin{remark}
In \cite{GDL21}, the authors gave the range of $\delta$ for BCH codes $\mathcal{C}_{(3,n,\delta})$ being dually-BCH codes  and showed that it looks much harder to give a characterisation of $\mathcal{C}_{(q,n,\delta})$ being dually-BCH codes, where $q>3$ is a prime power. Theorem \ref{theorem5} finished this work.
\end{remark}

\section{Conclusion}\label{sec-finals}
Let $n=\frac{q^m-1}{q+1}$ for $m\geq 4$ being even and $q$ being a prime power, or $n=\frac{q^m-1}{q-1}$ for $m\geq 4$ being a positive integer and $q$ being an odd prime power.
The main contributions of this paper are the following:
\begin{itemize}
\item Sufficient and necessary conditions for BCH codes $\mathcal{C}_{(q,n,\delta})$ being dually-BCH codes were given, where $2\leq \delta\leq n$. Lower bounds on the minimum distances of their dual codes are developed for $n=\frac{q^m-1}{q+1}$. In this sense, we extended the results in \cite{GDL21}.
\item We determined the largest coset leader modulo $\frac{q^m-1}{q-1}$, which is very useful to completely solve Open Problem 45 in \cite{Li2017}.
We also determined the largest coset leader modulo $\frac{q^m-1}{q+1}$, so the conjecture in \cite{WLP19} were proved.
\end{itemize}


\begin{thebibliography}{99}

\bibitem{Aly07} S. A. Aly, A. Klappenecker, P. K. Sarvepalli,  On quantum
and classical BCH codes, IEEE Trans. Inf. Theory 53 (3) (2007) 1183--1188.

\bibitem{Augot94} D. Augot, N. Sendrier,  Idempotents and the BCH bound, IEEE
Trans. Inf. Theory  40 (1) (1994) 204--207.

\bibitem{Berlekamp68} E. R. Berlekamp, Algebraic Coding Theroy, McGraw-Hill Book, New York (1968).

\bibitem{Berlekamp681} E. R. Berlekamp, Negacyclic codes for the Lee metric, In: Proceedings of the Conference on
Combinatorial Mathematics and Its Applications, pp. 298-316. Univerity of North Carolina Press, Chapel Hill (1968).

\bibitem{Bose62} R. C. Bose, D. K. Ray-Chaudhuri, On a class of error correcting binary group codes, Information and Control  3 (1962) 279-290.


\bibitem{Charpin90} P. Charpin, On a class of primitive BCH-codes, IEEE Trans. Inf.
Theory 36 (1) (1990) 222--228.

\bibitem{Dingbook15} C. Ding, Codes from Difference Sets, World Scientific, Singapore, 2015.

\bibitem{Ding2016} C. Ding, BCH codes in the past 55 years, The 7th international Workshop on Finite Fields Applications, Tianjin, China, 2016.

\bibitem{Ding2015} C. Ding, Parameters of several classes of BCH codes, IEEE Trans. Inf. Theory 61 (10) (2015) 5322-5330.

\bibitem{Ding15} C. Ding, X. Du, Z. Zhou, The Bose and minimum distance
of a class of BCH Codes, IEEE Trans. Inf. Theory 61 (5) (2015) 2351--2356.

\bibitem{Ding17} C. Ding, C. Fan, Z. Zhou, The dimension and minimum distance of two
classes of primitive BCH codes, Finite Fields Appl. 45 (2017) 237--263.




\bibitem{GDL21}
B Gong, C. Ding, C. Li, The dual codes of several classes of BCH codes, IEEE
Trans. Inf. Theory  68 (2) (2022)  953--964.

\bibitem{Gorenstein61} D. C. Gorenstein, N. Zierler, A class of error-correcting codes in $p^m$ symbols, J. SIAM 9 (1961) 207-214.

\bibitem{Grassl2006} M. Grassl, Bounds on the minimum distance of linear codes and quantum codes, Online available
at http://www.codetables.de.


\bibitem{Guo20} G. Guo, R. Li, Y. Liu, J. Wang, A family of negacyclic BCH codes of length $n=\frac{q^{2m}-1}{2}$, Crypto. Commun. 12 (2020) 187-203.

\bibitem{Hocquenghem59} A. Hocquenghem, Codes correcteurs d'erreurs, Chiffres (Pairs)  2 (1959) 147-156.

\bibitem{Kai13} X. Kai, S. Zhu, New quantum MDS codes from negacyclic codes, IEEE Trans. Inf. Theory 59 (2) (2013) 1193-1197.

\bibitem{Kai131} X. Kai, S. Zhu, Y. Tang, Quantum negacyclic codes, Phys. Rev. A 88 (1) 012326 (2013).

\bibitem{KS90}
A. Krishna, D. V. Sarwate, Pseudocyclic maximum-distance-separable codes,
IEEE Trans. Inf. Theory  36 (4) (1990) 880--884.

\bibitem{Lid017} C. Li, C. Ding, S. Li, LCD cyclic codes over finite fields, IEEE Trans.
Inf. Theory 63 (7) (2017) 4344--4356.

\bibitem{Li2017} S. Li, C. Ding, M. Xiong, G. Ge, Narrow-sense BCH codes over $\gf(q)$ with length $n=\frac{q^m-1}{q-1}$,
IEEE Trans. Inf. Theory 63 (11) (2007) 7219--7236.

\bibitem{Lid17} S. Li, C. Li, C. Ding, H. Liu, Two families of LCD BCH codes,
IEEE Trans. Inf. Theory 63 (9) (2017) 5699--5717.

\bibitem{Liu17} H. Liu, C. Ding, C. Li, Dimensions of three types of BCH codes over $\gf(q)$,
Discrete Math. 340 (2017) 1910--1927.

\bibitem{Li2019} Y. Liu, Y. Li, Q. Fu, L. Lu, Y. Rao, Some binary BCH codes with length $n=2^m+1$, Finite Fields Appl. 55 (2019) 109-133.

\bibitem{Liuv2} Y. Liu, R. Li, L. Guo, H. Song, Dimensions of nonbinary antiprimitive BCH codes and
some conjectures, arXiv: 1712.06842v2.


\bibitem{MS77}
F. J. MacWilliams, N. J. A. Sloane, \emph{The Theory of Error-Correcting Codes},
North-Holland, Amsterdam,  1977.




\bibitem{Pang18} B. Pang, S. Zhu, Z. Sun, On LCD Negacyclic Codes over Finite Fields, J. Syst. Sci. Complex 31 (2018) 1065--1077.

\bibitem{WLP19}
P. Wu, C. Li, W. Peng, On some cyclic codes of length $\frac{q^m-1}{q+1}$, Finite Fields Appl., 60 (2019) 101581.1-28

\bibitem{Rouayheb2007} S.Y. EI Rouayheb,  C. N. Georghiades, E. Soljanin, A. Sprintson, Bounds on codes based on graph theory, IEEE Int. Symp. on Information Theory. Nice, France  (2007) 1876-1879.



\bibitem{Yan2018} H. Yan, H. Liu, C. Li, S. Yang, Parameters of LCD BCH codes with two lengths, Adv. Math. Commun. 12 (3)
(2018) 579--594.

\bibitem{Yue15} D. Yue, Z. Feng, Minimum cyclotomic coset representatives
and their applications to BCH codes and Goppa codes, IEEE Trans.
Inf. Theory 46 (7) (2000) 2625--2628.

\bibitem{Dianwu96} D. Yue, Z. Hu, On the dimension and minimum
distance of BCH codes over $\gf(q)$, J. Electron. 13 (3) (1996) 216--221.

\bibitem{Zhu21} H. Zhu, M. Shi, X. Wang, T. Helleseth,  The $q$-ary antiprimitive BCH codes, IEEE Trans. Inf. Theory 68 (2022) 1683--1695.


\bibitem{Zhu20} S. Zhu, Z. Sun, P. Li, A class of negacyclic BCH codes and its application to quantum codes, Des. Codes. Crptogr.  86 (2018) 2139--2165.

\bibitem{Zhu19} S. Zhu, Z. Sun,  X. Kai,
A Class of Narrow-Sense BCH Codes. IEEE Trans. Inf. Theory 65 (2019) 4699--4714.









%
%
%

%
%
%



%

















\end{thebibliography}
\end{document}